\documentclass[acmsmall,nonacm]{acmart}
\setcopyright{none}

\usepackage{graphicx}

\usepackage[utf8]{inputenc}
\usepackage[OT2,T1]{fontenc}
\usepackage[english]{babel}
\usepackage{amsmath,amsfonts,cmll,bussproofs}
\usepackage{tikz}
\usepackage{hyperref}
\usepackage{enumitem}
\usepackage{bigints}
\usepackage{tikz}
\usepackage{color}
\usepackage{cleveref}
\usepackage{aliascnt}
\usepackage{xspace}

\theoremstyle{theorem}
\newtheorem{theorem}{Theorem}
\newtheorem{proposition}[theorem]{Proposition}

\theoremstyle{definition}
\newtheorem{definition}[theorem]{Definition}

\theoremstyle{remark}

\newtheorem*{remark}{Remark}
\newtheorem{example}[theorem]{Example}

\newcommand{\seq}[1]{#1}
\newcommand{\word}[1]{#1}
\newcommand{\countones}[2]{\#_{#1}(#2)}

\newcommand{\naturalN}{\mathbb{N}}
\newcommand{\abs}[1]{\mathopen{|}#1\mathclose{|}}

\newcommand{\pickedout}[2]{#2[\seq{#1}]}

\newcommand{\cylinder}[1]{[#1]}

\newcommand{\bigre}{{\sc BiGRE}\xspace}
\newcommand{\lobe}{{\sc LoBE}\xspace}
\newcommand{\cohop}{{\sc CoHOp}\xspace}
\newcommand{\dysco}{{\sc DySCo}\xspace}

\numberwithin{equation}{section}

\begin{document}

\title{Agafonov's Theorem for finite and infinite alphabets and probability distributions different from equidistribution}



\author{Thomas Seiller}
\authornote{T. Seiller was partially supported by the \grantsponsor{EC}{European Commission}{Horizon 2020 programme} Marie Sk\l{}odowska-Curie Individual Fellowship
(H2020-MSCA-IF-2014) project \grantnum{EC}{659920} - ReACT, the \grantsponsor{CNRS}{CNRS}{INS2I} grants \bigre and \lobe, the \grantsponsor{Ile-de-France}{DIM RFSI}{Exploratory project} Exploratory project \cohop, and the \grantsponsor{ANR}{ANR}{http://dx.doi.org/10.13039/501100001665} ANR-22-CE48-0003-01 project \dysco.}          
\orcid{0000-0001-6313-0898}             
\affiliation{
  \department[0]{LIPN -- UMR 7030 CNRS \& University of Paris 13}              
  \institution{CNRS}            
  \streetaddress{99, avenue Jean-Baptiste Clément}
  \city{Villetaneuse}
  \postcode{93430}
  \country{France}                    
}
\email{seiller@lipn.fr}   

\author{Jakob Grue Simonsen}
\orcid{0000-0002-3488-9392}             
\affiliation{
  \department[0]{Department of Computer Science (DIKU)}              
  \institution{University of Copenhagen}            
  \streetaddress{Universitetsparken 5}
  \city{København Ø}
  \postcode{2100}
  \country{Denmark}                    
}
\email{simonsen@di.ku.dk}   

%
%

\begin{abstract}
An infinite sequence $\alpha$ over an alphabet $\Sigma$ is $\mu$-distributed w.r.t. a probability map $\mu$ if, for every finite string $w$, the limiting frequency of $w$ in $\alpha$ exists and equals $\mu(w)$. 
We prove the following result for any finite or countably infinite alphabet $\Sigma$: every finite-state selector over $\Sigma$ selects a $\mu$-distributed sequence from every $\mu$-distributed sequence \emph{if and only if} $\mu$ is induced by a Bernoulli distribution on $\Sigma$, that is a probability distribution on the alphabet extended to words by taking the product. The primary -- and remarkable -- consequence of our main result is a complete characterization of the set of probability maps, on finite and infinite alphabets, for which finite-state selection preserves $\mu$-distributedness. The main positive takeaway is that (the appropriate generalization of) Agafonov's Theorem holds for Bernoulli distributions (rather than just equidistributions) on both finite and countably infinite alphabets. As a further consequence, we obtain a result in the area of symbolic dynamical systems: the shift-invariant measures $\mu$ on $\Sigma^{\omega}$ such that any finite-state selector preserves the property of genericity for $\mu$, are exactly the positive Bernoulli measures.
\end{abstract}

\maketitle


\section{Introduction}

Let $\alpha = x_1 x_2 \cdots$ be an infinite sequence over a finite alphabet $\Sigma$. A string $w \in \Sigma^*$ is said to occur in $\alpha$ with limiting frequency $f$ if
$\lim_{N \rightarrow \infty} \frac{\#_w(x_1 \cdots x_N)}{N} = f$, where $\#_w(x_1 \cdots x_N)$ is the number of times that $w$ occurs as a contiguous subsequence in $x_1 \cdots x_N$.
$\alpha$ is said to be \emph{normal} if every finite string of length $n$ over $\Sigma$ occurs with limiting
frequency $\vert \Sigma \vert^{-n}$ in $\alpha$  \citep{Borel}. By standard results, the fractional part of the base-$b$ expansion of almost all real numbers is a normal sequence for $b \geq 2$, 
so for base $10$, almost all real numbers have the digit ``0'' occurring 1-in-10 times in all sufficiently long finite prefixes of their digit expansion, have ``11'' occurring 1-in-100 times, 
``110'' occurring 1-in-1000 times, and so on. Concrete examples of normal sequences include Champernowne's sequence $1234567891011\cdots$ \citep{doi:10.1112/jlms/s1-8.4.254},
the Copeland-Erd{\"os} sequence $235711131719\cdots$ consisting of concatenating the prime numbers \citep{copeland1946}, and for any polynomial $f$ with positive integer coefficients the sequence $f(1)f(2)f(3) \cdots$ \citep{Davenport52noteon}. 

A \emph{finite-state selector}
is a DFA that selects those symbols $x_m$ from $\alpha$ such that $x_1 \cdots x_{m-1}$ is accepted by the DFA. The sequence of selected symbols may thus be finite or infinite.
\emph{Agafonov's Theorem} states that a sequence $\alpha$ is normal if{f} any DFA that selects an infinite sequence from $\alpha$, selects a normal sequence.
Colloquially, Agafonov's Theorem can be stated as: ``any constant-space algorithm must preserve normality''. 

The purpose of this paper is twofold: (I) we study whether analogues of Agafonov's Theorem holds if the distribution of finite strings is \emph{different} from equidistribution, i.e. whether distributions where finite strings $s$ are allowed
to occur with frequency distinct from $\vert \Sigma \vert^{-\vert w \vert}$; and (II) we study extensions of Agafonov's Theorem to infinite alphabets (which in the traditional setup in Agafonov's Theorem is meaningless as there is no equidistributed
probability distribution on a countably infinite set).

As an example, consider the (non-normal) sequence $\alpha = 010101\cdots$. Clearly, every finite bit string occurs in $\alpha$ with some well-defined frequency (the simplest way to see this is that for each $n > 0$, there are exactly two distinct
substrings of length $n$ in $\alpha$: one starting with $0$ and one starting with $1$), and the frequencies thus induce a probability distribution on $\{0,1\}^n$ for each $n$. In particular $0$ and $1$ each occur with limiting frequency $1/2$, but any DFA that selects symbols at even positions will select the sequence $111\cdots$, and thus the probability distribution on $\{0,1\}$ is \emph{not} preserved, showing that Agafonov's Theorem in general fails to hold. 

In addition to being intrinsically interesting, our study of Agafonov's Theorem is motivated by the fact that constant-space algorithms are usually employed in reactive programming languages used for signal processing (see Section \ref{sec:rel_other} below), both for transduction and selection, and Agafonov's Theorem is a strong guarantee that such algorithms will always preserve one notion of randomness for infinite strings, namely that the probability of a random length-$n$ subsequence being equal to a fixed word is exactly $\vert \Sigma \vert^{-n}$ -- as the above example shows, selection from sequences where $0$ and $1$ are known to occur with probability $1/2$ is not enough -- stronger guarantees such as normality must hold. Conversely, normality is a very strong requirement; in some infinite sequences, certain element may occur with much higher frequency than others, and one tantalizing way of generating new sequences having the same distribution of finite subsequences could be to simply let a DFA select elements from the original sequence, which in general is only possible if (the appropriate analogue) of Agafonov's Theorem holds.

The motivation for studying infinite alphabets is that the study of normality is closely tied to the study of symbolic dynamics and (information-theoretic) coding theory \citep{10.2307/26273928,Blanchard1992GenericST,lind_marcus_1995,Madritsch2018}, and that both areas have witnessed
recent advances using infinite alphabets \citep{DBLP:journals/tit/BoucheronGG09,6979852,6655916,7541049,MadritchMance:constructing}, in particular the techniques of Madritsch and Mance \citep{MadritchMance:constructing} have allowed
construction of Champernowne-like sequences for various distributions over infinite alphabets.

\subsection{Contribution}

The formal statement of the main theorem can be found in Theorem \ref{the:main} below. In plain language, we prove that:

\vspace{\baselineskip}

\emph{
Let $\Sigma$ be a non-empty finite or countably infinite alphabet, and let $\mu : \Sigma^* \longrightarrow [0,1]$ be a probability map (i.e., for all $n \geq 0$, $1 = \sum_{w \in \Sigma^n} \mu(w)$)
such that there exists at least one $\alpha \in \Sigma^\omega$ that is $\mu$-distributed. Then, the following are equivalent:
\begin{enumerate}
\item $\mu$ is induced by a positive Bernoulli probability distribution $p$ on $\Sigma$, i.e. for every $a_1,\ldots,a_n \in \Sigma$, $\mu(a_1 \cdots a_n) = \prod_{i=1}^n p(a)$, and for every $a \in \Sigma$, $p(a) > 0$.
\item For every DFA $A$ over $\Sigma$ and every $\mu$-distributed sequence $\alpha \in \Sigma^\omega$, if $A$ selects an infinite sequence from $\alpha$, then the selected sequence is $\mu$-distributed.
\end{enumerate}
}

\vspace{\baselineskip}

The above result completely characterizes the probability maps preserved by selection by DFAs, both for finite and infinite alphabets, and Agafonov's Theorem follows immediately as a corollary. We briefly review the 
roadmap and techniques used for the proof of the main result in Section \ref{sec:overview}.

As the study of distributions associated to limiting frequencies of finite strings in (right-)infinite strings is cryptomorphic to the study of shift-invariant probability measures on 
the shift space $(\Sigma^\omega,s)$ equipped with the $\sigma$-algebra induced by the basis of cylinder sets on $\Sigma$, we obtain as a corollary a result in the field of symbolic dynamical systems, namely a complete characterization of the shift-invariant
probability measures $\nu$ for which any finite-state selector preserves genericity for $\nu$, see Section \ref{sec:app_dyn}.

\subsection{Related work}

\subsubsection{Agafonov's Theorem and its generalizations}

Agafonov's Theorem \citep{Agafonov} was one of the end results of multiple efforts grappling with the two notions of (i) \emph{kollektiv} (roughly,
$\alpha \in \{0,1\}^\omega$ is a kollektiv wrt.\ a set $\mathcal{S}$ of selection strategies if the limiting frequency of $1$ is unchanged after applying any strategy in $\mathcal{S}$ to $\alpha$\footnote{The exact definition of kollektiv differs subtly across different authors, compare e.g. \citep{vonMiseskollektiv}, \citep{Church40}, and \citep{Postnikova61}. The original notion of kollektiv introduced by von Mises \citep{vonMiseskollektiv} had no constraints on the set $\mathcal{S}$, but this turned out to be essentially fruitless \citep{Tornier29,Reichenbach32,Kamke33,Copeland36}.}),
and (ii) \emph{admissible sequence} and its relation to the notion of normal sequence \citep{Copeland28,Reichenbach32,Reichenbach37,PostnikovPyateskii,Postnikov}. Agafonov's Theorem itself had a virtually unknown precursor in a beautiful result by Postnikova \citep{Postnikova61} that showed, with the terminology of the present paper, that $\alpha \in \{0,1\}^\omega$ is normal if{f} the distribution of 1s is preserved by selection strategies depending only on a finite word (see \autoref{def:strategy} for the formal definition of \emph{Postnikova strategies}).

%
%
%
%
%
%

Both Postnikova \citep{Postnikova61} and Agafonov \citep{AgafonovRussianLong} considered selection functions on sequences in $\{0,1\}^\omega$ where the limiting distribution of $1$ was $0 < p < 1$ (i.e., considered a Bernoulli distribution on $\{0,1\}$), but considering Bernoulli distributions instead of the special case of equidistributions seems to have disappeared almost completely from all later work. One possible reason for this is that only the short version (without proofs or explanation of techniques) of Agafonov's result \citep{Agafonov} appeared in English as \citep{AgafonovSummary}; in contrast, the original longer paper in Russian \citep{AgafonovRussianLong} was published in a more obscure journal, and was never translated. We have provided a (very) embellished account of the arguments in \citep{AgafonovRussianLong} on the preprint server \emph{arXiv}\footnote{\href{https://arxiv.org/abs/2007.03249}{https://arxiv.org/abs/2007.03249}.} where we expand Agafonov's terse use of existing results of the time in much more detail, including using more basic arguments with modern methods (e.g., using concentration bounds directly instead of appealing to the law of large numbers) and add further embellishments to Agafonov’s original arguments. Many of the results in the present paper exist due to insights obtained due to this embellishment, rather than the original proof itself.

For equidistribution, the earliest extension to arbitrary alphabets seems to be by Broglio and Liardet \citep{BroglioLiardet}, and a number of authors have since re-proved Agafonov's Theorem in the special case of equidistribution using a variety of methods; for example, using predictors defined from finite automata (for $\Sigma = \{0,1\}$) \citep{OCONNOR1988324}, using compressibility arguments \citep{BECHER2013109,BECHER20151592,DBLP:conf/fct/Shen17}, and a combination of automata-theoretic and 
probabilistic methods similar to Agafonov's original reasoning \citep{Cartonfinite}.

Agafonov's Theorem itself has been generalized to treat selectors that are not necessarily (induced by prefix selection by) finite automata \citep{airey2015,BECHER20151592,Vandehey:uncanny,CartonVandehey},
and some generalizations consider selectors based on relaxed finiteness criteria of the syntactic monoid of a language selecting prefixes of infinite sequences \citep{KamaeWeiss,KamaeWang}; conversely, results by Merkle and Reimann show that adding just slight computational power to the selection strategies beyond finite automata -- e.g. using a Pushdown automaton with unary stack alphabet instead of a DFA\footnote{In fact, one of the strategies considered by Merkle and Riemann, which consists in computing the language $\{ww^{\mathrm{R}}\mid w\in\Sigma^*\}$ where $w^{\mathrm{R}}$ is the \emph{reverse} of $w$, can be computed by an arguably less expressive model of computation, namely two-way automata with two heads \citep{holzer_multi-head_2008}.} renders Agafonov's Theorem invalid \citep{DBLP:journals/mst/MerkleR06}. Similarly, selection by finite automata has been extended, and analogues for Agafonov's Theorem been proved, in other settings than selection from elements
of the set $\Sigma^\omega$, e.g. for shifts of finite type \citep{Cartonfinite}.  All of these results only consider normality rather than more general classes of distributions on finite strings.

Conversely, \emph{construction} of normal sequences (as opposed to selecting normal sequences from other normal ones) has been investigated thoroughly for more than a hundred years
\citep{Sierpinski1917,doi:10.1112/jlms/s1-8.4.254,YoshinobuNakai1992,Normalitydigitarticle,Mance2012CantorSC,Pollackarticle}, including explicit construction of real numbers with normal expansion for any integer base $b \geq 2$ \citep{Levin:early,Scheererarticle,Aisarticle}, and real numbers
with normal expansion in non-integer bases \citep{VANDEHEY2016424,Madritscharticle}. Among this work, the result of most use to the present paper is the construction by Madritsch and Mance of generic sequences for any shift-invariant probability measure $\mu$ \citep{MadritchMance:constructing} -- these are essentially sequences that are $\mu$-distributed using the terminology of the present paper (see Definition \ref{def:el_mu_dist}).

In very recent work, Carton \citep{Cartonfinite} proves that, for any Markov measure $\mu$ on $\Sigma^\omega$ induced by a pair $(P,\pi)$ of a stochastic $\vert \Sigma \vert \times \vert \Sigma \vert$ matrix and a stationary distribution $\pi$ for $P$, any sequence selected from a $\mu$-distributed sequence by a finite-state selector from a particular subset of $\mu$-\emph{compatible} selectors, will be $\mu$-distributed. Roughly, a finite-state selector is compatible, if it can only read consecutive symbols of $\Sigma$ with non-zero transition probability in $P$ and every state has only incoming transitions of at most
one symbol from $\Sigma$. In contrast, we consider the full set of finite-state selectors. Moreover, Carton's results are restricted to the case of finite alphabets.

\subsubsection{Streams and selection from infinite sequences}\label{sec:rel_other}

Infinite streams are typically used to model situtations where data elements arrive, no upper bound on the length of the stream is known a priori, and the focus
is not on resource use as a function of the length of the stream; for example, infinite streams have been studied extensively
in event-level differential privacy \citep{DBLP:conf/soda/Dwork10,DBLP:journals/pvldb/KellarisPXP14}, and in semantics of lazy programming languages such as \textsc{Haskell} \citep{Haskell}.

\emph{Selection} of (substreams of) elements from infinite streams has been investigated from a practical perspective since the 1960s \citep{DBLP:journals/acta/Stephens97}, and is typically performed by specialized stream processing languages, e.g.\ \textsc{LUSTRE} \citep{DBLP:conf/popl/CaspiPHP87}
and \textsc{ESTEREL} \citep{BERRY199287}, typically for use in reactive programming (e.g., for signal processing or circuit design). As they are designed for real-time processing, these languages typically allow only very constrained operations -- any program in both \textsc{LUSTRE} and \textsc{Esterel} can be compiled to a finite state transducer automaton (and deterministic program selecting a subsequence from its input is hence a finite-state selector as in Agafonov's Theorem).

In typical algorithmic treatments of stream processing, one typically studies unordered, finite sequences of elements from a very large, or infinite, set \citep{Muthukrishnan:streams}. The problems considered typically have strong constraints, e.g. that only a single pass over the stream is allowed and that each element can only be observed once, and often involve a \emph{sketch}--a data structure that stores information about the elements seen in the stream and allows to answer predefined queries. A classic example is estimating the frequency moments of the distribution of elements in the stream using sketches with low memory in both alphabet size and stream length \citep{ALON1999137,DBLP:conf/stoc/IndykW05,10.1007/978-3-642-40328-6_5}. Our work can be seen as a variation of streaming where the alphabet size may be infinite, the stream itself is infinite, and the distribution of element is not limited to the set of elements, but also has requirements on the finite subsequences of elements in the stream; in this setting, our main result is that any constant-space sketch sampling an infinite stream in real-time preserves the distribution of finite subsequences if{f} the distribution is induced by a Bernoulli distribution on the set of elements.

\subsection{Overview of techniques and the proof of the main theorem}\label{sec:overview}

The main result has two directions: (I) proving that if $\mu$-distributedness is preserved by selection by any DFA, then $\mu$ is necessarily induced by a Bernoulli distribution, and (II) any $\mu$ induced by a Bernoulli distribution is preserved across selection by any DFA.

For (I), we prove the more general result that if $\mu$ is not induced by a Bernoulli distribution on $\Sigma$,
selection by a particular \emph{Postnikova strategy} (roughly, a Postnikova strategy selects an element of the sequence if and only if it follows a fixed finite word) will select a non-$\mu$-distributed infinite sequence from a -- bespoke -- $\mu$-distributed sequence. The Postnikova strategy contains prefixes in the form $u\cdot w \in \Sigma^*$ for a \emph{fixed} $w$ chosen such that $w \cdot a \in \Sigma^{\vert w \vert+1}$ is a minimal witness string such that $\mu(w \cdot a) \neq \mu(w) \cdot \mu(a)$. Using basic constructions, we can then prove that the Postnikova strategy can be implemented by a DFA that simulates a sliding fixed-width window.

For (II), most of the modern methods of proving Agafonov's Theorem (e.g., \citep{BECHER2013109,BECHER20151592,DBLP:conf/fct/Shen17}) are not immediately adaptable because they use methods that are particular to equidistributions on finite alphabets (e.g., lossless finite-state compressors \citep{BECHER2013109} or automatic Kolmogorov complexity \citep{DBLP:conf/fct/Shen17}) -- and we consider both Bernoulli distributions and infinite alphabets. Instead, we work along the general lines of Agafonov's original proof \citep{Agafonov} that more heavily uses probabilistic reasoning. 

The key insights in Agafonov's original proof was (i) that any \emph{strongly connected} finite automaton (containing at least one accepting state) applied to a normal sequence must select (always, not just with probability $1$) more than a constant fraction of elements from any sufficiently long finite substring of its input, and (ii) that selecting more than a constant fraction of sufficiently long substrings entails that each element of $ \Sigma$ must be selected with approximately equal probability, by the Law of Large Numbers. 
In Agafonov's original approach (for $\Sigma = \{0,1\}$), an appeal to the Strong Law of Large Numbers was used in conjunction with the product measure on
the product topology on $\{0,1\}^\omega$, thus required reasoning about cylinder sets centered on sets $A$ of finite strings; and to avoid ``double-counting'' the probabilities, these sets had to be prefix-free. We avoid this difficulty by
using concentration bounds to tally the occurrences of elements $a \in \Sigma$ in block decompositions of finite prefixes of $\alpha$. 

The proof that any DFA selects a $\mu$-distributed infinite sequence from a $\mu$-distributed infinite sequence
then follows by observing that (i) any run of a DFA on an infinite sequence eventually reaches a strongly connected component $C$ of the DFA that is recurrent (i.e., the run can never exit $C$), and (ii) that any such component induces
an irreducible Markov chain, whence we can apply the Ergodic Theorem for Markov Chains to conclude that accepting states are reached infinitely often and with appropriate frequency.

The extension to infinite alphabets is surprisingly straightforward in most proofs: essentially, instead of using combinatorial estimates for finite sets, we have
to ensure that series taken over infinite alphabets converge properly, but almost all instances involve series that (i) have non-negative elements, and (ii) are bounded above, whence the usual reasoning about absolutely convergent series can be employed. Similarly, the classic results for finite automata that we use need to be re-stated and re-proved in the case of infinite alphabets, but this in general turns out to be doable without too much leg-work (e.g. Lemma \ref{lem:aux_ss}). One caveat is that several important ancillary results
have standard proofs that use combinatorial arguments on finite sets, and we thus need to provide alternative proofs using different methods.

\section{Preliminaries}

\begin{definition}
We assume a non-empty, possibly (countably) infinite, alphabet $\Sigma$ and denote by $\lambda$ the empty string; the sets of finite and right-infinite sequences of elements of $\Sigma$
are denoted by $\Sigma^*$ and $\Sigma^\omega$, respectively. We denote by $\Sigma^+$ the set of finite non-empty words, i.e. $\Sigma^+=\Sigma^*\setminus\{\lambda\}$. Elements of $\Sigma^*$ are ranged over by
$\word{v}, \word{w},\ldots$, and elements of $\Sigma^\omega$ by $\alpha, \beta, \ldots$.
If $\alpha = a_1 a_2 \cdots \in \Sigma^\omega$ and $N$ is a positive integer, we denote by
$\alpha \vert_{\leq N}$ the finite string $a_1 a_2 \cdots a_N$. 

Given $v \in \Sigma^*$ and $u \in \Sigma^* \cup \Sigma^\omega$, we write $v\cdot w$ for the element of $\Sigma^* \cup \Sigma^\omega$ obtained by concatenation. For words $v \in \Sigma^*$ and $w \in \Sigma^* \cup \Sigma^\omega$, $v$ is said to be a \emph{prefix} of $w$, written $v \preceq w$, if there exists $u \in \Sigma^* \cup \Sigma^\omega$ such that
$w = v\cdot u$. If $v \preceq w$ and $v \neq w$, $v$ is said to be a \emph{proper prefix} of $w$, written $v \prec w$. 
For any $\word{w} \in \Sigma^*$,
the \emph{cylinder set} of $\word{w}$, denoted $\cylinder{\word{w}}$, is the subset of $\Sigma^\omega$ defined by
$\cylinder{\word{w}} = \{\alpha \in \Sigma^\omega : \alpha = w \cdot \beta, \beta \in \Sigma^\omega\}$, that is the set of right-infinite sequences
that have $\word{w}$ as prefix.
\end{definition}

\begin{definition}
Let $\Sigma$ be a non-empty, possibly (countably) infinite, alphabet.
A \emph{probability map} (over $\Sigma$) is a map $\mu : \Sigma^+ \longrightarrow [0,1]$
such that, for all positive integers $n$, the series 
$$
\sum_{a_1 \cdots a_n \in \Sigma^n} \mu(a_1 \cdots a_n)
$$
\noindent is convergent with limit $1$. Note that convergence implies absolute convergence here.

A probability map $\mu$ is said to be:

\begin{itemize}

\item  \emph{induced by a Bernoulli distribution} $p : \Sigma \longrightarrow [0,1]$ if, for all positive integers $n$,
and all $a_1,\ldots,a_n \in \Sigma$,
$\mu(a_1 \cdots a_n) = \prod_{i=1}^n \mu(a_i) = \prod_{i=1}^n p(a_i)$.

\item \emph{invariant} if, for all $\word{w} \in \Sigma^*$ the series 
$\sum_{a \in \Sigma} \mu(\word{w}\cdot a)$ and $\sum_{a \in \Sigma} \mu(a \cdot \word{w})$
are convergent with limit $\mu(\word{w})$.

\item (when $\Sigma$ is finite) \emph{equidistributed} if, for any
$\word{w} \in \Sigma^{n}$,
 $\mu(\word{w}) = \vert \Sigma \vert^{-n}$.
\end{itemize}
\end{definition}

Observe that an equidistributed $\mu$ is also Bernoulli.
For alphabets $\vert \Sigma \vert > 1$,
any map $p : \Sigma \longrightarrow
[0,1]$ such that the series $\sum_{a \in \Sigma} p(a)$ converges to $1$
induces a probability map $\mu_p$ by
setting $\mu_p(a_1 \cdots a_n) = \prod_{j=1}^n p(a_j)$.
For finite alphabets $\Sigma$, this map is equidistributed if{f} $p(a) = \vert \Sigma \vert^{-1}$
for every $a \in \Sigma$.

The expression ``induced by a Bernoulli distribution'' is justified by the fact that Bernoulli probability maps
correspond directly to the measure of cylinders in Bernoulli shifts \citep{shields:bernoulli}\footnote{In the literature
on normal numbers, the word Bernoulli is sometimes used slightly differently, for example Schnorr and Stimm \citep{DBLP:journals/acta/SchnorrS72}
use the term ``Bernoulli sequence'' for sequences that are equidistributed in our terminology. We also note that $\mu$-distributed sequences (defined on the next page)
 w.r.t. Bernoulli distributions were first introduced by Postnikov and I. I. Piatetski-Shapiro under the name ``Bernoulli normal sequences'' \citep{BernoulliNormal}.}.

\begin{proposition}\label{prop:Bernoulli_is_invariant}
A probability map $\mu$ induced by a Bernoulli distribution is invariant.
\end{proposition}

\begin{proof}
For any $w \in \Sigma^*$, $\sum_{a\in \Sigma} \mu(a w) = \sum_{a \in \Sigma} \mu(a)\mu(w) = \sum_{a \in \Sigma} \mu(w)\mu(a) =
\sum_{a \in \Sigma} \mu(wa)$. And $\sum_{a \in \Sigma} \mu(w)\mu(a) = \mu(w)\sum_{a \in \Sigma} \mu(a) = \mu(w)$.
\end{proof}

We shall need probability maps to act as ``measures'' on (possibly infinite) sets of finite strings:

\begin{definition}
Let $\Sigma$ be a non-empty alphabet, let $W \subseteq \Sigma^*$, and let $\mu$ be a probability map over $\Sigma$. If $W = \emptyset$, we define $\mu(W) = 0$.
If $\sum_{w \in W} \mu(w)$ converges, we define $\mu(W) =  \sum_{w \in W} \mu(w)$.
\end{definition}

Observe that as $\mu(w) \geq 0$ for all $w \in \Sigma^*$, if $\sum_{w \in W} \mu(w)$ converges, it is absolutely convergent (hence, we do not need to specify an ordering
of $W$).

We are interested in the probability maps whose values can be realized as the limiting frequencies 
of finite words in right-infinite sequences over $\Sigma$.

\begin{definition}\label{def:el_mu_dist}
Let $\word{v} = v_1 \cdots v_N$ and $\word{w} = w_1 \cdots w_n$ be finite words over $\Sigma$. We denote by
$\#_{\word{w}}(\word{v})$ the number of occurrences of $\word{w}$ in $\word{v}$, that is, the quantity
$$
\left\vert \left\{j\leq N+1-n : v_j v_{j+1} \cdots v_{j+n-1} = w_1 w_2 \cdots w_n \right\}\right\vert
$$

Let $\mu$ be a probability map over $\Sigma$, and let $\alpha$ be a right-infinite sequence over $\Sigma$.
If the limit
$$
\lim_{N \rightarrow \infty} \frac{\#_{\word{w}}(\alpha\vert_{\leq_N})}{N}
$$
exists and is equal to some real number $f$, we say that $\word{w}$ occurs in $\alpha$ \emph{with limiting frequency} $f$.
If every $\word{w} \in \Sigma^+$ occurs in $\alpha$ with limiting frequency $\mu(\word{w})$,
we say that $\alpha$ is $\mu$-distributed.
\end{definition}

\begin{proposition}\label{prop:invariancefromgenericity}
Let $\mu$ be a probability map over $\Sigma$. If there exists a $\mu$-distributed sequence, then $\mu$ is invariant.
\end{proposition}

\begin{proof}
Let $\mu$ be a probability map over $\Sigma$ and $\alpha=a_1 a_2\dots $ a $\mu$-distributed sequence. We consider $w=w_1 w_2\dots w_k \in \Sigma^k$ and note that for all $N > 0$:
\[ \left\vert \sum_{a \in \Sigma} \#_{wa}(\alpha\vert_{\leq N})- \#_{w}(\alpha\vert_{\leq N}) \right\vert \leq 1. \]
Indeed, every occurence of $w$ as $a_i a_{i+1} \dots a_{i+k}$ such that $i>1$ is also an occurence of $b\cdot w$ for a (unique) $b\in\Sigma$, so the expressions $\#_{w}(\alpha\vert_{\leq N})$ and $\sum_{a \in \Sigma} \#_{wa}(\alpha\vert_{\leq N})$ are equal if and only if $a_1 a_{2} \dots a_{k}\neq w$ and their difference is equal to $1$ otherwise.

Thus 
\[ \left\vert \frac{\sum_{a \in \Sigma} \#_{wa}(\alpha\vert_{\leq N})}{N}- \frac{\#_{w}(\alpha\vert_{\leq N})}{N}\right\vert \leq \frac{1}{N}. \]

We therefore obtain that:
\[ \left\vert  \frac{\sum_{a \in \Sigma} \#_{wa}(\alpha\vert_{\leq N})}{N} - \mu(w) \right\vert \leq \left\vert  \frac{\sum_{a \in \Sigma} \#_{wa}(\alpha\vert_{\leq N})}{N} - \frac{\#_{w}(\alpha\vert_{\leq N})}{N} \right\vert + \left\vert \frac{\#_{w}(\alpha\vert_{\leq N})}{N} - \mu(w) \right\vert. \]
Since both expressions on the right converge to $0$, the left-hand side converges to zero, showing that $\mu(w)=\lim_{n\rightarrow\infty} \frac{\sum_{a \in \Sigma} \#_{wa}(\alpha\vert_{\leq n})}{n}=\sum_{a\in\Sigma} \lim_{n\rightarrow\infty}\frac{\#_{wa}(\alpha\vert_{\leq n})}{n}=\sum_{a\in\Sigma} \mu(wa)$.

Similarly, for all $N>0$:
\[ \left\vert \sum_{a \in \Sigma} \#_{wa}(\alpha\vert_{\leq N})-\#_{w}(\alpha\vert_{\leq N})\right\vert \leq 1, \]
by a similar argument as the one used above, noting that the number of occurrences is different if and only if $a_{N-k+1}a_{N-k+2}\dots a_N= w$. We then conclude that $\mu(w)=\sum_{a\in\Sigma} \mu(aw)$ in the same way.
\end{proof}

Observe that an infinite sequence $\alpha$ is normal in the usual sense if{f} it is $\mu$-distributed for (the unique) equidistributed
probability map $\mu$ over $\Sigma$. Also observe that it is not all probability maps $\mu$ for which there
exists a $\mu$-distributed sequence.

\begin{example}
An example of a probability map that is \emph{not} Bernoulli,
but such that there is at least one $\mu$-distributed right-infinite sequence,
is the map $\mu$ over $\Sigma = \{0,1\}$ defined by $\mu(w) = 1/2$ if $w$ does not contain any of the
strings $00$ or $11$ (note that for each positive integer $n$, there
are exactly two such strings of length $n$, namely $0101010\cdots$ and $101010\cdots$), and
$\mu(w) = 0$ otherwise. Observe that the right-infinite sequence
$010101 \cdots$ is $\mu$-distributed.
\end{example}


In contrast to all previous work on Agafonov's Theorem, we allow countably infinite alphabets $\Sigma$. Alphabets of larger cardinality
do not in general have probability measures realizable by considering limiting frequencies of elements of $\Sigma^\omega$ -- simply because
most elements of $\Sigma$ cannot occur at all in a single element of $\Sigma^\omega$.

One reason why previous generalizations of Agafonov's Theorem have not considered infinite alphabets is that there can be no equidistribution on a countably
infinite set. However, there are Bernoulli measures $\mu$ on countably infinite alphabets $\Sigma$ and $\mu$-distributed infinite sequences over $\Sigma$.

\begin{example}
An example of a countably infinite alphabet with a Bernoulli measure is $\Sigma = \mathbb{N}$
and $p(n) = 6/(\pi n)^2$ (note that we have $\sum_{n \in \Sigma} p(n) = 1$). In general, any convergent series $\sum_{n=1}^\infty a_n$ where every $a_n$ is non-negative induces
a Bernoulli distribution on $\mathbb{N}$ by setting $p(n) = a_n/ \sum_{n=1}^\infty a_n$.
Each such Bernoulli distribution $p$ induces an invariant probability map $\mu_p$, and by a result of Madritsch and Mance \citep{MadritchMance:constructing}, there exists
a $\mu_p$-distributed sequence.
\end{example}

\begin{remark}
As we consider possibly infinite alphabets, we often have to consider infinite series instead of finite sums in the proofs. In most cases, these series will have elements that are known to be non-negative,
and the sum of all partial sums will be bounded above, whence the series will be absolutely convergent and the order of summation can thus be changed freely. A trivial example of use is to consider some $B \subseteq \Sigma$
and note that $\sum_{a \in B} p(a) = 1 - \sum_{a \in \Sigma \setminus B} p(a)$ (as $\sum_{a \in B} p(a) \leq 1$, $\sum_{a \in \Sigma \setminus B} p(a) \leq 1$, and $p(a) \geq 0$ the two series are absolutely convergent,
and $\sum_{a \in B} p(a)  +  \sum_{a \in \Sigma \setminus B} = \sum_{a \in \Sigma} p(a) = 1$). 
\end{remark}



\subsection{Strategies}

\begin{definition}\label{def:strategy}
Let $\Sigma$ be an alphabet.
A \emph{strategy} $S$ over $\Sigma$ is a subset $S\subseteq \Sigma^{*}$.

Given a strategy $S$ and $\alpha \in \Sigma^\omega$, we define the sequence \emph{selected} by $S$, denoted $S[\alpha]$, as follows: if $i_{1},i_{2},\dots,i_{k},\dots$ is the (increasing) sequence of indices $i_j$ such that $\alpha\vert_ {< i_j} \in S$, then $S[\alpha]_{j}= \alpha_{i_{j}}$. When $w \in \Sigma^*$ is a finite word, we define $S[w]$ \emph{mutatis mutandis}.

A strategy $S$ is a \emph{Postnikova strategy} if there is $w \in \Sigma^*$ such that $S = \Sigma^* w$.
\end{definition}

Thus, $S[\alpha]$ is simply the subsequence
of symbols from $\alpha$ that are ``picked out'' by applying $S$ 
to prefixes of $\alpha$. Note also
that if $\word{w} \in S$, then 
in any word on the form
$\word{w} \cdot b \cdot \word{v}$, 
$S$ must pick $b$. Thus, $S$ cannot be made
to, for instance, only pick out a single symbol from $\Sigma$ -- it must select ``the next symbol'' after any
$\word{w} \in S$. This precludes, for example, constant-memory strategies from selecting only $0$s from a normal binary sequence.

Our primary object of study is the case where $S$ is a regular language, described next.

\subsection{Finite-State Selectors and selection by DFAs}

As we treat both finite and (countably) infinite alphabets, we must consider automata over possibly infinite alphabets. Every automaton has a finite number of states as usual,
but as the alphabet is infinite and a deterministic automaton has transitions on all symbols from every state, the underlying graph of the automaton will be infinitely branching.
To keep notations simple, we refer to deterministic automata with a finite number of states as ``DFA''s as usual, even if the underlying alphabet is infinite.
 
 \begin{definition}
 A \emph{finite-state selector} over $\Sigma$ is a DFA
 $A = (Q,\delta,q_s,Q_F)$, where $Q$ is the set of states,
$q_s$ is the unique start state, $Q_F$ is the set of accepting states,
and $\delta : Q \times \Sigma \longrightarrow Q$ is the transition relation.

A DFA is strongly connected
 if its underlying directed graph (states are nodes, transitions are edges) is strongly connected.

  Denote by $L(A)$ the language accepted by the automaton.
 If $\alpha = a_1 a_2 \cdots $ is a finite or right-infinite sequence over $\Sigma$, the 
 subsequence \emph{selected by} $A$ is the (possibly empty) sequence of letters
 $a_n$ such that the prefix $a_1 \cdots a_{n-1} \in L(A)$, that is,
 the automaton when started on the finite word $a_1 \cdots a_{n-1}$ in state $q_s$
ends in an accepting state after having read the entire word. The \emph{run} of $A$ on input $\alpha$ is the sequence
of states visited when $A$ is applied to $\alpha$ from the starting state. For $(q,w) = (q,w_1 \cdots w_n) \in Q \times \Sigma^*$, we use the notation
$\delta^*(q,w)$ to denote the state $\delta(\cdots \delta(\delta(q,w_1),w_2) \ldots w_n)$, that is, the state reached by starting from  $q$ and following the (unique)
path induced by $w$. 
 
\end{definition}

Observe that a DFA may select an empty, finite or infinite sequence when run on a right-infinite word.

\begin{definition}
Let $A$ be a DFA. A strongly connected component $C$  in (the underlying directed graph of) $A$ is said to be \emph{recurrent} if, for every state $p$ in $C$ and every
$a \in \Sigma$,$\delta(p,a)$ is a state in $C$ (i.e., once a run of $A$ on some infinite word reaches a state in $C$, the run cannot leave $C$).
\end{definition}

\begin{definition}
Let $A = (Q,\Sigma,\delta,q_0,F)$ be a connected DFA. For all $q\in Q$, we denote by $A_{q}$ the automaton
$(Q,\Sigma,\delta,q,F)$,
i.e. where the state $q$ is chosen as the initial state.
\end{definition}

\begin{definition}
Let $A = (Q,\Sigma,\delta,q_0,F)$ be a connected DFA, and let $q\in Q$. Let $\alpha$ be a right-infinite sequence over $\Sigma$. We denote by $\pickedout{\alpha}{A_{q}}$ the subsequence $\bar{\alpha}$ of $\alpha$ \emph{picked out} by $A_{q}$, that is, $w_{i}\in\bar{\word{w}}$ if and only if $A_{q}(\word{w}_{<i})$ reaches an accepting state.
\end{definition}

We shall use the following fundamental result in automata theory\footnote{The result in \citep{DBLP:journals/acta/SchnorrS72} is stated for finite alphabets, but the proof method carries through for infinite alphabets as well. We provide
a proof in Appendix \ref{sec:auxiliary}.}:

\begin{lemma}[Lemma 2.6 of \citep{DBLP:journals/acta/SchnorrS72}]\label{lem:SS}
For every DFA $A =  (Q,\delta,q_s,Q_F)$ over (the possibly infinite) alphabet $\Sigma$, there is a 
word $w \in \Sigma^*$ such that, for every $q \in Q$,  there is a strongly connected recurrent component $C$ of (the underlying directed graph of $A$) such that $\delta^*(q,w) \in C$.
\end{lemma}

\begin{corollary}\label{cor:strong_components_are_enough}
Let $\mu$ be a probability map induced by a positive Bernoulli distribution on $\Sigma$, let $A$ be a DFA over $\Sigma$, and let $\alpha \in \Sigma^\omega$ be $\mu$-distributed. Then, the run of $A$ on $\alpha$ eventually reaches a strongly connected recurrent component of $A$.
\end{corollary}

\begin{proof}
Let $w$ be the word obtained from Lemma \ref{lem:SS}.
As $\alpha$ is $p$-distributed, $w$ appears in $\alpha$, so write $\alpha = v w \alpha'$, and let $q$ be the 
state of $A$ reached after $\vert v \vert$ transitions in the run of $A$ on $\alpha$. Then, after at most a further $\vert w \vert$ transitions, the run reaches a state in a strongly connected component of (the underlying directed graph of) $A$.
\end{proof}

Corollary \ref{cor:strong_components_are_enough} ensures that we can assume without loss of generality that the finite-state selectors we treat are strongly connected. Note that the corollary does not imply that the strongly connected recurrent component contains an accepting state (indeed, the automata may have an empty set of accepting states). Thus, some automata do not always select infinite sequences, and additional assumptions are needed if this is desirable (this is discussed in Remark \ref{rem:2} below). However, this is not an issue for our main result which states that the output of a selector applied to a normal sequence is again normal \emph{as long as it is infinite}.

\section{Main result}

\begin{theorem}\label{the:main}
Let $\Sigma$ be a non-empty (finite or infinite) alphabet and $\mu$ be a probability map such that there exists at least one $\alpha \in \Sigma^\omega$ that is $\mu$-distributed.
Then, the following statements are equivalent:
\begin{enumerate}
\item \label{mainthm:bernoulli} $\mu$ is induced by a positive Bernoulli distribution $p$ on $\Sigma$, that is, for every $a_1 \cdots a_n \in \Sigma$, $\mu(a_1 \cdots a_n) = \prod_{i=1}^n  p(a)$, and $p(a) > 0$ for all $a \in \Sigma$;
\item \label{mainthm:postnikova} (Postnikova property) for every finite word $w\in\Sigma^{\ast}$ and $\mu$-distributed sequence $\alpha \in \Sigma^\omega$, if the sequence selected from $\alpha$ by the Postnikova strategy $\mathcal{S}_w=\{ u\in \Sigma^*\mid \exists v\textrm{ s.t. }u=v w\}$ is infinite, then it is $\mu$-distributed;
\item \label{mainthm:agafonov} (Agafonov property) For every DFA $A$ over $\Sigma$ and every $\mu$-distributed sequence $\alpha \in \Sigma^\omega$, if the sequence selected from $\alpha$ by $A$ is infinite, then it is $\mu$-distributed.
\end{enumerate}
\end{theorem}


\begin{proof}
For the implication \ref{mainthm:bernoulli} $\Rightarrow$ \ref{mainthm:agafonov}, \autoref{cor:strong_components_are_enough} yields that any run of a finite-state selector on a $\mu$-distributed sequence eventually reaches 
a strongly connected recurrent component; the restriction of any DFA to the state set of one of its recurrent component is 
also a DFA, and the result now follows by \autoref{lem:final_crucial}. The implication \ref{mainthm:agafonov} $\Rightarrow$ \ref{mainthm:postnikova} is clear from the definitions since the considered strategies are computed by finite automata (\ref{postnikovastrategiesbyautomata}). Lastly, \autoref{lem:crucialPostnikova} and \autoref{lem:non_positive_Post} prove that \ref{mainthm:postnikova} $\Rightarrow$ \ref{mainthm:bernoulli}.
\end{proof}

\begin{remark}\label{rem:2}
Theorem \ref{the:main} addresses the case where a DFA or Postnikova strategy selects an infinite sequence from a $\mu$-distributed sequence. If one wants to restrict attention to automata that
\emph{always} select an infinite subsequence from any $\mu$-distributed sequence, extra conditions sometimes occur in the literature, e.g. that every cycle in the (underlying graph of the) DFA contains an accepting
state \citep{BECHER2013109} ensuring that an infinite subsequence is selected from any (not just $\mu$-distributed sequence). Another condition that ensures that an infinite subsequence is selected from
any $\mu$-distributed sequence is to consider only DFAs such that every strongly connected recurrent component contains at least one accepting state. In this case, Corollary \ref{cor:strong_components_are_enough}
ensures that any run on the automaton on a $\mu$-distributed sequence will reach a strongly recurrent component, and Lemma \ref{main:lemma1} below then ensures that the DFA accepts an infinite subsequence from $\alpha$.
\end{remark}

\section{Non-preservation of $\mu$-distributedness for non-Bernoulli measures}

We first prove that if $\mu$ is a probability map such that any DFA selects a $\mu$-distributed
right-infinite sequence from any $\mu$-distributed right-infinite sequence, then $\mu$ must be Bernoulli. This is 
an immediate consequence of a stronger property proved in Lemma \ref{lem:crucialPostnikova} below.

The idea of the proof is that if $\mu$ is \emph{not} Bernoulli, there exists a word $a_1 \cdots a_k$ such that
$\mu(a_1 \cdots a_{k-1}) = \prod_{j=1}^{k-1} a_j$, but
$\mu(a_1 \cdots a_{k-1} a_k) \neq \mu(a_1 \cdots a_{k-1}) \cdot \mu(a_k)$. One can then
construct a finite-state selector that acts like a ``sliding window'' of size $k-1$, that is, remembers the last $k-1$ letters scanned and accepts
if these are $a_1 \cdots a_k$. This selector will select every letter following $a_1 \cdots a_{k-1}$; after a prefix of length $N$ of a right-infinite
sequence has been scanned, approximately $N \cdot \mu(a_1 \cdots a_{k-1})$ have been selected, and approximately
$N \cdot \mu(a_1 \cdots a_{k-1} a_k)$ of these will be the symbol $a_k$. But then the limiting frequency of $a_k$ in the sequence selected 
will be $\mu(a_1 \cdots a_{k-1} a_k)/\mu(a_1 \cdots a_{k-1}) \neq \mu(a_k)$, and the result follows. 

For completeness, we give a fully formal proof after the lemma, but the entirety of the reasoning is essentially as we just described.

%

\begin{lemma}\label{lem:crucialPostnikova}
Let $\mu  : \Sigma^* \longrightarrow [0,1]$ be a probability map.
If $\mu$ is not induced by a   Bernoulli distribution on $\Sigma$, there exists a finite word $w\in \Sigma^*$ such that if
$\alpha \in \Sigma^\omega$ is $\mu$-distributed, then 
the Postnikova strategy $\mathcal{S}_w=\{ u\in \Sigma^*\mid \exists v\textrm{ s.t. }u=vw\}$ 
selects from $\alpha$ an infinite sequence $\beta \in \Sigma^\omega$ that is \emph{not} 
$\mu$-distributed.
\end{lemma}

\begin{proof}
If no element of $\Sigma^\omega$ is $\mu$-distributed, the lemma is vacuously true. 
Hence, assume that there is at least one $\alpha \in \Sigma^\omega$ that is $\mu$-distributed. If $\vert \Sigma \vert = 1$, then there is exactly one probability map on $\Sigma^*$, namely
the one that assigns probability $1$ to the unique element of $\Sigma^k$ for every  $k \geq 0$, and this probability map is clearly Bernoulli, and the lemma is thus vacuously true. Hence, in the remainder of the proof,
assume that $\vert \Sigma \vert \geq 2$.

Assume that $\mu$ is not induced by a Bernoulli distribution on $\Sigma$. Then there are $k$ and a word  $a_1\cdots_{k-1} a_k \in \Sigma^k$
such that $\mu(a_1 \cdots a_{k-1} a_k) \neq \prod_{j=1}^k \mu(a_j)$. Observe that $k = 1$ is impossible
and thus we must have $k \geq 2$. Assume without loss of generality that
$k$ is minimal among such $k$, and hence that
$\mu(a_1 \cdots a_{k-1}) = \prod_{j=1}^{k-1} \mu(a_j)$, and note that this implies
$\mu(a_1 \cdots a_{k-1} a_k) \neq \mu(a_1 \cdots a_{k-1}) \cdot \mu(a_k)$.

Assume for contradiction that $\mu(a_1 \cdots a_{k-1}) = 0$. Then $\mu(a_i) = 0$ for at least one $a_i$ 
and thus $\mu(a_1 \cdots a_{k-1} a_k) = 0$, because the fact that 
there is at least one $\mu$-distributed right-infinite sequence entails that $\mu(a_1 \cdots a_{k-1} a_k) > 0$ implies $\mu(a_i) \geq \mu(a_1 \cdots a_{k-1} a_k) > 0$.
But this is a contradiction as we would then have $\mu(a_1 \cdots a_{k-1} a_k) = 0 = \prod_{j=1}^k \mu(a_j)$.
Thus, $\mu(a_1 \cdots a_{k-1}) > 0$.

As  $\mu(a_1 \cdots a_{k-1} a_k) \neq \mu(a_1 \cdots a_{k-1}) \cdot \mu(a_k)$,  $\mu(a_1 \cdots a_{k-1}) > 0$, and we have
$\mu(a_1 \cdots a_{k-1} a_k) \leq \mu(a_1 \cdots a_{k-1})$ (because $\mu$ is invariant by Proposition \ref{prop:invariancefromgenericity}),
there is a real number $\gamma$ with
$0 < \gamma < 1$ such that:
$$
\left\vert \frac{\mu(a_1 \cdots a_{k-1} a_k)}{ \mu(a_1 \cdots a_{k-1})} -  \mu(a_k) \right\vert > \gamma 
$$

We now consider the Postnikova strategy $\mathcal{S}_w$ with $w=a_1 \cdots a_{k-1}$, i.e. the strategy that selects exactly the symbols following the occurrences of $a_1 \cdots a_{k-1}$ in $\alpha$.

Let $\alpha \in \Sigma^\omega$ be $\mu$-distributed. Then, 
for every $\epsilon > 0$, there is an $N_\epsilon > 0$
such that for all $n >N_\epsilon$ we have:
\begin{equation*}
\left\vert\frac{\countones{a_1 \cdots a_{k-1}}{\alpha\vert_{\leq n}}}{n}-\mu(a_1 \cdots a_{k-1})\right\vert\leq\epsilon
\end{equation*}
Hence
\begin{equation}\label{eq:fjol_1}
n \mu(a_1 \cdots a_{k-1}) - n\epsilon \leq \countones{a_1 \cdots a_{k-1}}{\alpha\vert_{\leq n}} \leq n \mu(a_1 \cdots a_{k-1}) + n\epsilon
\end{equation}
\noindent and
\begin{equation}\label{eq:fjol_2}
n \mu(a_1 \cdots a_{k-1} a_k) - n\epsilon \leq \countones{a_1 \cdots a_{k-1} a_k}{\alpha\vert_{\leq n}} \leq n \mu(a_1 \cdots a_{k-1} a_k) + n\epsilon
\end{equation}
As $\mu(a_1 \cdots a_{k-1}) > 0$ and $\mathcal{S}_w$ selects the symbol after each occurrence of $a_{1} \cdots a_{k-1}$, $\mathcal{S}_w$ selects an infinite sequence $\beta$ from $\alpha$.
%
%
%
Let $\beta^{(n)} \in \Sigma^*$ be the finite sequence selected by $\mathcal{S}_w$ from $\alpha\vert_{\leq n}$. Observe that we have
$\vert \beta^{(n)} \vert = \countones{a_1 \cdots a_{k-1}}{\alpha\vert_{< n}}$,
and $\countones{a_k}{\beta^{(n)}} = \countones{a_1 \cdots a_{k-1}a_k}{\alpha\vert_{< n}}$.
The fraction of occurrences $\countones{a_k}{\beta^{(n)}}/\vert \beta^{(n)} \vert$ of $a_k$ in $\beta^{(n)}$ thus satisfies:
$$
\frac{\countones{a_k}{ \beta^{(n)}}}{\vert \beta^{(n)} \vert} = \frac{ \countones{a_1 \cdots a_{k-1}a_k}{\alpha\vert_{< n}}}{ \countones{a_1 \cdots a_{k-1}}{\alpha\vert_{< n}}}
= \frac{\countones{a_1 \cdots a_{k-1}a_k}{\alpha\vert_{< n}}}{n} \cdot \frac{n}{ \countones{a_1 \cdots a_{k-1}}{\alpha\vert_{< n}}}
$$
and hence, by (\ref{eq:fjol_1}) and (\ref{eq:fjol_2}), for all $n > N$:
\begin{equation}\label{eq:squeeze_k}
\frac{\mu(a_1 \cdots a_{k-1} a_k) - \epsilon}{\mu(a_1 \cdots a_{k-1}) + \epsilon} \leq \frac{\countones{a_k}{ \beta^{(n)}}}{\vert \beta^{(n)} \vert} \leq
\frac{\mu(a_1 \cdots a_{k-1} a_k) + \epsilon}{\mu(a_1 \cdots a_{k-1}) - \epsilon}
\end{equation}
Consider an arbitrary $\delta$ with $0 < \delta < \gamma/2$. By (\ref{eq:squeeze_k}), for all sufficiently small $\epsilon$, we have 
$$
\left\vert  \frac{\mu(a_1 \cdots a_{k-1}a_k)}{\mu(a_1 \cdots a_{k-1})} - \frac{\countones{a_k}{ \beta^{(n)}}}{\vert \beta^{(n)} \vert} \right\vert < \delta
$$
\noindent and thus for all $n > N_\epsilon$:
 \begin{align*}
\gamma < \left\vert \frac{\mu(a_1 \cdots a_{k-1} a_k)}{\mu(a_1 \cdots a_{k-1})} - \mu(a_k) \right\vert &\leq \left\vert  \frac{\mu(a_1 \cdots a_{k-1} a_k)}{\mu(a_1 \cdots a_{k-1})} - \frac{\countones{a_k}{ \beta^{(n)}}}{\vert \beta^{(n)} \vert} \right\vert +  \left\vert \frac{\countones{a_k}{ \beta^{(n)}}}{\vert \beta^{(n)} \vert} -  \mu(a_k)   \right\vert \\
& < \delta +  \left\vert \frac{\countones{a_k}{ \beta^{(n)}}}{\vert \beta^{(n)} \vert} -  \mu(a_k)   \right\vert < \frac{\gamma}{2} + \left\vert \frac{\countones{a_k}{ \beta^{(n)}}}{\vert \beta^{(n)} \vert} -  \mu(a_k)   \right\vert
\end{align*}
\noindent whence:
$$
 \left\vert \frac{\countones{a_k}{ \beta^{(n)}}}{\vert \beta^{(n)} \vert} -  \mu(a_k)   \right\vert > \frac{\gamma}{2}
$$
\noindent and as the sequence $(\beta^{(n)})_{n \in \mathbb{N}}$ consists of prefixes of the sequence $\mathcal{S}_w[\alpha]$ selected by $\mathcal{S}_w$ from $\alpha$, and is eventually increasing, the frequency
of occurrences of $a_k$ differs infinitely often from $\mu(a_k)$ by at least $\gamma/2$,  $\mathcal{S}_w[\alpha]$ cannot be $\mu$-distributed.
\end{proof}

Lemma \ref{lem:crucialPostnikova} shows that if a probability map is not induced by a Bernoulli distribution on $\Sigma$, some Postnikova strategy will select a non-$\mu$-distributed sequence from \emph{any} $\mu$-distributed sequence. In case $\mu$ \emph{is} induced by a Bernoulli distribution, but not a positive Bernoulli distribution, we can show the weaker result that there will be a Postnikova strategy that selects a non-$\mu$-distributed sequence form \emph{some} $\mu$-distributed sequences
(and this is sufficient for our main Theorem).

\begin{lemma}\label{lem:non_positive_Post}
Let $\mu  : \Sigma^* \longrightarrow [0,1]$ be a probability map induced by a Bernoulli distribution on $\Sigma$ that is not positive.
Then there exists a finite word $w\in \Sigma^*$ and $\mu$-distributed
$\alpha \in \Sigma^\omega$ such that
the Postnikova strategy $\mathcal{S}_w=\{ u\in \Sigma^*\mid \exists v\textrm{ s.t. }u=vw\}$ 
selects from $\alpha$ an infinite sequence $\beta \in \Sigma^\omega$ that is not 
$\mu$-distributed.
\end{lemma}

\begin{proof}
As $\mu$ is not positive, pick $b \in \Sigma$ such that $\mu(b) = 0$, and let $\Gamma \subseteq \Sigma$ be a maximal subset such that the restriction of $\mu$ to $\Gamma$
is a positive Bernoulli distribution (observe that $\Gamma$ is non-empty because $\mu$ is a probability map and $1 = \sum_{a \in \Sigma} \mu(a)$ thus implies $\mu(a) > 0$ for some $a \in \Sigma$).
By \citep{MadritchMance:constructing} there exists a $\mu$-distributed infinite sequence $\beta \in \Gamma^\omega$; notice that $\beta$ can be assumed w.l.o.g. to not contain any occurrences of $b$. Let $\alpha \in \Sigma^\omega$ be obtained by inserting  the string $bb$ at positions
$2,4,8,16,\ldots$. Then, $\alpha$ is $\mu$-distributed because (i) every $v \in \Gamma^*$ occurs with the same limiting frequency as in $\beta$\footnote{The key observation here is that since the 'bb's are inserted at exponentially increasing positions, the frequency of occurrence of all other strings is decreased by a very small (and quickly decaying) factor.}, and every $v \in \Sigma^*$ that contains
an element of $\Sigma\setminus \Gamma$ occurs in $\alpha$ with limiting frequency $0$. Set $w = b$; then the Postnikova strategy
$\mathcal{S}_w = \{u \in \Sigma^* \vert \exists v \mathrm{ s.t. } u = vw\}$ selects from $\alpha$ a sequence $\beta = \mathcal{S}_w[\alpha]$ such that, for every $n > 0$, 
$\#_{b}(\beta\vert_{\leq n}) \geq n/2 - 1$. Thus, the limiting frequency of $b$ in $\beta$ is not $0$, and hence is not $\mu(b)$, proving
that $\beta$ is not $\mu$-distributed.
\end{proof}

\begin{lemma}\label{postnikovastrategiesbyautomata}
Let $w\in\Sigma^*$. The Postnikova strategy $\{ u\in \Sigma^*\mid \exists v\textrm{ s.t. }u=vw\}$ is computable by a strongly connected DFA over $\Sigma$.
\end{lemma}

\begin{proof}
Note that the alphabet can possibly be infinite in the following proof. In the trivial case $\vert \Sigma \vert = 1$, the result is trivial since there is only one infinite sequence $\alpha$ and every Postnikova strategy extracts $\alpha$ from $\alpha$. We therefore now suppose that the alphabet is of size at least $2$. 

We write $m$ the length of the word $w$, and write $w_1,w_2,\dots,w_m$ the bits of $w$. We design a finite state selector $M_w$ with exactly $2^m$ states which will select a bit of the input if and only if it is preceded by the word $w$. Let $M_w=(\mathbf{2^m},\delta,q_s,Q_F)$ be defined as follows:
\begin{itemize}
\item $\mathbf{2^m}=\{(b_1,b_2,\dots,b_m) : b_i\in\{0,1\}\}$ is the set of binary sequences of length $m$; those will represent a sequence of  bits where $b_j=1$ if and only if the previous $j$ bits of the input coincide with the first $j$ bits of the input;
\item $q_s$ the initial state is chosen to be the sequence $(0,0,\dots,0)\in \mathbf{2^m}$;
\item $Q_F$ the set of accepting states is equal to the set of sequences $\{(b_1,b_2,\dots,b_m)\in \mathbf{2^m} : b_m=1\}$;
\item $\delta$ the transition function is defined as $\delta(b_1,b_2,\dots,b_m;a)=(c_1,c_2,\dots,c_m)$ where $c_j=1$ if and only if $b_{j-1}=1$ and $a=w_j$ for $j\neq 1$, and $c_1=1$ if and only $a=w_1$.
\end{itemize}

The fact that this automaton computes the Postnikova strategy is clear from the definition. We now show it is strongly connected by showing that any state $(b_1,b_2,\dots,b_m)$ is reachable from an arbitrary state. For this, we consider a word $u_{b_1,b_2,\dots,b_m}=u_1,\dots,u_m$ defined by $u_i = w_i$ if and only if $b_i=1$ (and thus $u_i\neq w_i$ whenever $b_i=0$ -- which we can chose since the alphabet contains at least two symbols). We then claim that the automaton, starting from any state $c\in \mathbf{2^m}$, reaches the state $(b_1,b_2,\dots,b_m)$ when given the word $u_{b_1,b_2,\dots,b_m}$ as input.
\end{proof}

\section{Finite-state selectors preserve $\mu$-distributedness for Bernoulli measures}

The sequence of auxiliary results of this section follows the general lines of Agafonov's original proof in Russian for the case $\Sigma = \{0,1\}$ \citep{Agafonov}, but with multiple proofs needing more careful analysis and adapted techniques.

\subsection{Ancillary definitions and results}

\begin{definition}
Let $\Sigma$ be an alphabet, $\alpha = x_1 x_2 \cdots x_n \cdots\in \Sigma^\omega$, and let $n$ be a positive integer. 
The $n$-\emph{block decomposition of} $\alpha$ is the sequence
$(\alpha_{(n,r)})_{r \geq 1}$ where $\alpha_{(n,r)} = x_{(r-1)n+1} \cdots x_{rn} \in \Sigma^n$.
\end{definition}

Thus, $\alpha_{(n,1)}$ is the string of the first $n$ symbols of $\alpha$, $\alpha_{(n,2)}$ is the string of the next $n$ symbols, and so forth.

\begin{definition}
Let $\mu$ be a probability map over $\Sigma$ and $\alpha=x_{1}x_{2}\cdots x_{n} \cdots \in \Sigma^\omega$. We say that $\alpha$ is $\mu$-\emph{block-distributed} if, for each $n \geq 1$ and every $\word{w} \in \Sigma^n$,
the $n$-block decomposition $(\alpha_{(n,r)})_{r\geq 1}$ of $\alpha$ satisfies:
$$
\lim_{k \rightarrow \infty} \frac{\vert i \leq k : \alpha_{(n,k)} = \word{w} \vert}{k} = \mu(\word{w})
$$
\end{definition}

For finite alphabets and the special case of $p$ being an equidistribution on $\Sigma$, it is straightforward to prove that the properties of being $\mu_p$-distributed and $\mu_p$-block-distributed are equivalent \citep{niven1951,cassels1952,Postnikova}. For the present paper, we only use that $\mu_p$-distributedness implies $\mu_p$-block-distributedness, which follows by tedious, but standard counting arguments on sufficiently large finite prefixes of $\alpha\vert_{\leq N}$ using
the same reasoning as the original proof by Niven and Zuckerman for finite alphabets and normality \citep{niven1951}, \emph{mutatis mutandis}:

\begin{proposition}\label{prop:mu_implies_mu}
Let $\mu_p$ be a probability map induced by a Bernoulli distribution $p$ on the alphabet $\Sigma$. If $\alpha \in \Sigma^\omega$
is $\mu_p$-distributed, it is $\mu_p$-block distributed.
\end{proposition}

We now prove that finite-state selectors can be composed appropriately; this will later be a key ingredient in reducing the problem of selecting finite strings $w \in \Sigma^*$ with frequency $\mu_p(w)$ to the problem of selecting single symbols $a \in \Sigma$ with frequency $p(a)$.

\begin{proposition}[Finite-State selectors are compositional]
\label{prop:fin_sel_comp}
Let $A$ and $B$ be DFAs over the same alphabet. Then there is a DFA $C$ such that, for each
sequence $\word{w}$, $\pickedout{\word{w}}{C} = \pickedout{\pickedout{\word{w}}{A}}{B}$.
If $A$ and $B$ are both strongly connected and $A$ contains at least one accepting state, $C$ can be chosen to be strongly connected.
\end{proposition}

\begin{proof}
Let $A = (Q^A,\Sigma,\delta^A,q_0^A,F^A)$
and $B = (Q^B,\Sigma,\delta^B,q_0^B,F^B)$.
Define $Q^C = Q^A \times Q^B$,
and set $q_0^C = (q_0^A,q_0^B)$
and $F^C = F^A \times F^B$.
For each $q^B \in Q^B$,
define the set $D_{q^B} = \{(q,q^B) : q \in Q^A \}
\subseteq Q^C$. Observe that
$Q^C = \bigcup_{q^B \in Q^B} D_{q^B}$ and
that for $q^B, r^B \in Q^B$ with
$q^B \neq r^B$, we have $D_{q^B} \cap D_{r^B} = \emptyset$, and thus $\{D_{q^B} : q^B \in Q^B\}$ is a partitioning of $Q^C$.
Hence, the transition relation, $\delta^C$, of $C$ may be defined by 
defining it separately on each subset $D_{q^B}$:
$$
\delta^C((q,q^B),a) = \left\{
\begin{array}{ll}
    (r,q^B) & \textrm{if } q \notin F^A  \textrm{ and } \delta^A(q,a) = r \\
     (r,r^B) & \textrm{if } q \in F^A \textrm{ and } \delta^A(q,a) = r
        \textrm{ and } \delta^B(q^B,a) = r^B \\

\end{array}
\right.
$$
Thus, when $C$ processes its input, it freezes the current state $q^B$
of $B$ (the freezing is represented by staying within $D_{q^B}$) and simulates
$A$ until an accepting state of $A$ is reached (i.e. just before $A$ would select the
next symbol); on the next transition, $C$ unfreezes the current state of $B$
and moves to the next state $r^B$ of $B$ and then freezes it and continues with a
simulation of $A$.

Observe that a symbol is picked out by $C$ if{f} the state is an element of $F^C = F^A \times F^B$ if{f} the symbol is the next symbol
read after simulation of $A$ reaches an accepting state of $A$ when
the current frozen state of $B$ is an accepting state of $B$.

By construction, $C$ is strongly connected if both $A$ and $B$ are: for any pair of states $(q_1^A,q_1^B)$ and $(q_2^A,a^2_B)$ in $Q^C$,
strong connectivity of $B$ implies that there is a directed path from $q_1^B$ to $q_2^B$ in $B$. 
Let $q_1^B, q_{1,2}^B, q_{1,3}^B,\ldots,q_{1,k}^B$ be the states along this path.
Strong connectivity of $A$ and the assumption that there is some $q_1^F \in F^A$ imply that 
there is a directed path from $(q_1^A,q_1^B)$ to $(q_1^F,q_1^B)$ in $C$, and by definition of $\delta^C$, there is a transition in $C$
from $(q_1^F,q_1^B)$ to $(q_1^B,q_{1,2}^B)$. A straightforward induction on $k$ now completes the proof.
\end{proof} 

%

The following shows that to prove that the property of being $ \mu_p$-distributed is preserved
under finite-state selection, it suffices to prove that the
limiting frequency of each $a \in \Sigma$ exists and equals $p(a)$.

\begin{lemma}\label{lem:simply_normal_is_enough}
Let $\mu_p$ be a probability map induced by a Bernoulli distribution $p$ on $\Sigma$, and let
$\alpha \in \Sigma^\omega$ be $\mu_p$-distributed. The following are equivalent:

\begin{itemize}

\item For all strongly connected DFAs $A$, if $A[\alpha]$ is infinite, then $A[\alpha]$ is $\mu_p$-distributed.

\item For all strongly connected DFAs $A$ and all $a \in \Sigma$, if $A[\alpha]$ is infinite, then
the limiting frequency of $a$ in $A[\alpha]$ exists and 
equals $p(a)$.

\end{itemize}
\end{lemma}

\begin{proof}
If, for all $A$ such that $A[\alpha]$ is infinite, $A[\alpha]$ is $\mu_p$-distributed, then in particular
the limiting frequency of $a$ in $A[\alpha]$ exists and 
is equal to $p(a)$ for all $A$.

Conversely, suppose that, for all strongly connected DFAs $A$ and all $a \in \Sigma$, 
if $A[\alpha]$ is infinite, then the limiting frequency of $a$ in $A[\alpha]$ exists and 
equals $p(a)$. If $\vert \Sigma \vert = 1$, it follows immediately that $A[\alpha]$ is $\mu_p$-distributed; hence, in the remainder of the proof, assume that $\vert \Sigma \vert \geq 2$. 

Let $A$ be a strongly connected DFA such that $A[\alpha]$ is infinite. 
If $A$ has no accepting states, there is nothing to prove, so assume that
$A$ has at least one accepting state.

We will prove by induction on $k \geq 0$
that the limiting frequency of 
every $v_1 \cdots v_k v_{k+1} \in \Sigma^{k+1}$ exists and equals
$\mu_p(v_1 \cdots v_k v_{k+1})$. 

\begin{itemize}
    \item $k = 0$: This is the supposition.
    
    \item $k \geq 1$. Suppose
    that the result has been proved
    for $k-1$. Let $v = v_1 \cdots v_k \in \Sigma^k$; by the induction hypothesis, the limiting frequency of $v_1 \cdots v_k$ in $A[\word{w}]$
    is $\mu_p(v_1 \cdots v_k)$. We claim that there is a strongly connected DFA
    $B$ that, from any sequence $\alpha$, selects the symbol after each occurrence
    of $v_1 \cdots v_k$, and \emph{only} those symbols, except for at most $k$ symbols at the start of $\alpha$. We construct such a DFA as folllows: Set $B = (\{q_0,\ldots,q_k\},\delta, q_0, \{q_k\}\}$, with $\delta$ to be defined below.
For $i$ with $0\leq i\leq k$, the state $q_i$ represents a situation where the last $i$ symbols read by $M_w$ is a length-$i$ prefix of $i$ among the last $k$ symbols read, and $i$ is maximal (i.e., there is no $j$ with $i < j \leq k$ such that the last $j$ symbols read by $B$ is also a prefix of $v$); observe that $v$ can overlap with itself, e.g. $wv= 000$, so when $B$ has read the string $100$, $i=2$, but the rightmost $0$ in $100$ is also a prefix of $v$).

We define $\delta$ as follows, for any $q_i$ and $a \in \Sigma$:

$\delta(q_i,a) = q_j$ where $j$ is the largest $j$ with $0 < j \leq k$ such that $v_{k-j+1} \cdots v_i \cdot a$ is a prefix of $v$. Note in particular that if $v_1 \cdots v_i a = v_1 \cdots v_i v_{i+1}$, then 
$\delta(q_i,a) = q_{i+1}$. If no such $j$ exists (i.e., no prefix of $v$ overlaps with the last $k$ symbols read), we define $\delta(q_i,a) = q_0$. Observe in particular, that $q_0$ has transitions to itself on all symbols $a$ such that $a \neq v_1$.

To see that $M_w$ is strongly connected, we prove the stronger property that between any (not necessarily distinct) states $q_i, q_j$, gthere is a path containing the state $q_k$. Observe that the state $q_i$ represents the situation where
the prefix $v_1 \cdots v_i$ has been read by $B$, and that there is a path $(q_{i},v_{i+1}) \cdots (q_{k-1}, v_{k})$ to $q_k$. As $\vert \Sigma \vert \geq 2$, there is at least one symbol $a \in \Sigma$ such that no suffix of $v_2 \cdots v_k \cdot a$ is a non-empty prefix of $v$, and hence there is at least one transition from $q_k$ to $q_0$; this proves strong connectivity.


    By \autoref{prop:fin_sel_comp},
    there is a strongly connected DFA $C$
    such that $C[\word{w}] = B[A[\word{w}]]$ for all $w \in \Sigma^*$.
    
    For any
    $a \in \Sigma$ and
    any sufficiently large positive integer $N$, we have:
    $$
    \frac{\countones{a}{\pickedout{\alpha\vert_{\leq N}}{C}}}{\vert \pickedout{\alpha\vert_{\leq N}}{C} \vert} =
    \frac{\countones{a}{\pickedout{\pickedout{\alpha\vert_{\leq N}}{A}}{B}}}{\vert \pickedout{\pickedout{\alpha\vert_{\leq N}}{A}}{B} \vert} =   
    \frac{\countones{v_1 \cdots v_k a}{\pickedout{\alpha\vert_{\leq N}}{A}}}{\countones{v_1 \cdots v_k }{\pickedout{\alpha\vert_{\leq N}}{A}}}
    $$

By the induction hypothesis,
for every $\epsilon > 0$,
we have, for all sufficiently large $N$, that
$\left\vert \frac{\countones{a}{\pickedout{\alpha\vert_{\leq N}}{C}}}{\vert \pickedout{\alpha\vert_{\leq N}}{C} \vert}
- p(a) \right\vert < \epsilon$, and hence:
\begin{equation}\label{eq:fur_first}
\left\vert\frac{\countones{v_1 \cdots v_k a}{\pickedout{\alpha\vert_{\leq N}}{A}}}{\countones{v_1 \cdots v_k }{\pickedout{\alpha\vert_{\leq N}}{A}}}
- p(a) \right\vert < \epsilon
\end{equation}

But for all sufficiently large $N$,
the induction hypothesis also furnishes that:
\begin{equation}\label{eq:fur_second}
\left\vert \frac{\countones{v_1 \cdots v_k }{\pickedout{\alpha\vert_{\leq N}}{A}}}{\vert \pickedout{\alpha\vert_{\leq N}}{A}\vert}
 - \mu_p(v_1 \cdots v_k)\right\vert < \epsilon
 \end{equation}
But as:
 $$
 \frac{\countones{v_1 \cdots v_k a}{\pickedout{\alpha\vert_{\leq N}}{A}}}{\vert \pickedout{\alpha\vert_{\leq N}}{A}\vert}
 =
 \frac{\countones{v_1 \cdots v_k a}{\pickedout{\alpha\vert_{\leq N}}{A}}}{\countones{v_1 \cdots v_k }{\pickedout{\alpha\vert_{\leq N}}{A}}}
 \cdot \frac{\countones{v_1 \cdots v_k }{\pickedout{\alpha\vert_{\leq N}}{A}}}{\vert \pickedout{\alpha\vert_{\leq N}}{A}\vert}
 $$
Equations \ref{eq:fur_first} and \ref{eq:fur_second} thus yield:
{\small
 \begin{align*}
 \left\vert \frac{\countones{v_1 \cdots v_k a}{\pickedout{\alpha\vert_{\leq N}}{A}}}{\vert \pickedout{\word{\alpha\vert_{\leq N}}}{A}\vert} - \mu_p(v_1 \cdots v_k a)\right\vert 
&= \left\vert \frac{\countones{v_1 \cdots v_k a}{\pickedout{\alpha\vert_{\leq N}}{A}}}{\countones{v_1 \cdots v_k }{\pickedout{\alpha\vert_{\leq N}}{A}}}
 \cdot \frac{\countones{v_1 \cdots v_k }{\pickedout{\alpha\vert_{\leq N}}{A}}}{\vert \pickedout{\alpha\vert_{\leq N}}{A}\vert}  - \mu_p(v_1 \cdots v_k)p(a) \right\vert \\
 &<
 \epsilon^2 + \epsilon\left( \frac{\countones{v_1 \cdots v_k a}{\pickedout{\alpha\vert_{\leq N}}{A}}}{\countones{v_1 \cdots v_k }{\pickedout{\alpha\vert_{\leq N}}{A}}}
 + \frac{\countones{v_1 \cdots v_k }{\pickedout{\alpha\vert_{\leq N}}{A}}}{\vert \pickedout{\alpha\vert_{\leq N}}{A}\vert}\right)\\
 &\leq \epsilon^2 + 2\epsilon
 \end{align*}
}
 Hence, for all $a \in \Sigma$, the limiting frequency of $v_1 \cdots v_k a$ in $\pickedout{\alpha}{A}$  exists and equals $\mu_p(v_1 \cdots v_k a)$,
 as desired.
\end{itemize}
\end{proof}

\subsection{Preservation of Bernoulli $\mu_p$-distributedness under finite-state selection}

By Lemma \ref{lem:simply_normal_is_enough} we may restrict our attention to proving that the frequency of single symbols from $\Sigma$
are preserved under selection by DFAs. The strategy will be to consider an arbitrary strongly connected DFA $A$, split the set of finite words $\Sigma^*$ into multiple classes that depend
on the selection behaviour of $A$, and use a combination of concentration bounds and basic Markov chain theory applied to these classes to obtain upper and lower bounds on the frequency with which $A$ selects
each symbol from $A$.



\begin{definition}
Let $A = (Q,\Sigma,\delta,q_0,F)$ be a strongly connected DFA. For any probability distribution $p : \Sigma \longrightarrow [0,1]$, any $b\in [0,1]$, $n\in\naturalN$, and any $\epsilon>0$, we define sets $D^p_n(b,\epsilon)$,
$E_n(b)$ and $G_n(b,\epsilon)$ as follows:
{\small
\begin{align*}
D_{n}^{p}(b,\epsilon,q) &= \left\{\word{w}\in\Sigma^{n}~:~  \vert\pickedout{\word{w}}{A_{q}} \vert >bn \textrm{ and } \sup_{a \in \Sigma} \left\vert \frac{\countones{a}{\pickedout{w}{A_{q}}}}{\vert\pickedout{w}{A_{q}}\vert}-p(a)\right\vert<\epsilon\right\} \\
D_n^p(b,\epsilon) &= \bigcap_{q \in Q} D_n^p(b,\epsilon,q)\\
E_n(b,q) &= \{\word{w}\in\Sigma^{n} : \vert A_{q}[\word{w}] \vert\leq bn\} \\
E_{n}(b) &= \bigcup_{q\in Q} E_{n}(b,q)\\
G_n(b,\epsilon,q) &= \left\{\word{w}\in\Sigma^{n} : \vert A_{q}[\word{w}]\vert >bn \textrm{ and } \sup_{a \in \Sigma} \left\vert\frac{\countones{a}{A_{q}[\word{w}]}}{\vert A_{q}[\word{w}] \vert}-p(a) \right\vert \geq\epsilon\right\}\\
G_{n}(b,\epsilon) &=\bigcup_{q\in Q}G_{n}(b,\epsilon,q)
\end{align*}
}
\end{definition}

Observe that, for all $b,n,\epsilon$,
$$
\Sigma^n = E_n(b) \cup D_n^p(b,\epsilon) \cup G_n(b,\epsilon)
$$
\noindent (and also note that $E_n(b)$ and $G_n(b,\epsilon)$ are not necessarily disjoint).

\begin{lemma}\label{main:lemma1}
Let $p$ be a positive Bernoulli distribution on $\Sigma$, and
let $S=(Q,\Sigma,\delta,q_s,Q_F)$ be a strongly connected finite automaton with $Q_F \neq \emptyset$, and let $n$ be a positive integer.
 Then there exists a real number $c >0 $ such that for all real numbers $\epsilon>0$ we have
$\lim_{n\rightarrow\infty} \mu_{p}\left( E_{n} \left(c-\epsilon \right)\right)=0$.
\end{lemma}

\begin{proof}


$S$ induces a stochastic
$\vert Q \vert \times \vert Q \vert$ matrix $\mathbf{P}$ 
by setting
$$\mathbf{P}_{ij} = \sum_{a \in \Sigma} 
p(a) \cdot 1_ {\delta(i,a)=j}.
$$
Observe that if $\Sigma$ is infinite, the fact that (i) $p(a) \cdot 1_{\delta(i,a)=j} \geq 0$, (ii)
$p(a) \cdot 1_{\delta(i,a)=j} \leq p(a)$, and (iii) $\sum_{a \in \Sigma} p(a) = 1$ entails
that the series $\sum_{a \in \Sigma} 
p(a) \cdot 1_{\delta(i,a)=j}$ is absolutely convergent. 

 Note also that
$\mathbf{P}_{ij} =0$ if{f} there are no transitions from
$i$ to $j$ in $Q$ on a symbol $a \in \Sigma$
with $p(a) > 0$. 
As $S$ is strongly connected,
there exists a path from state $i$ to state $j$ for each  $i,j \in Q$. Let $v$ be the word along this path; as $p(a) > 0$ for all $a \in \Sigma$, we have $\mu_p(v) > 0$, whence for each $i,j$ there is an integer
$n_{ij}$ such that $\mathbf{P}^{n_{ij}}_{ij} > 0$,
that is, $\mathbf{P}$ (and its associated Markov chain) is irreducible. As all states of a finite Markov chain
with irreducible transition matrix are positive recurrent, standard results (see, e.g., \cite[Thm.\ 54]{Serfozo:basics}) yield that there is a unique positive stationary distribution
$\pi : Q \longrightarrow [0,1]$
(s.t., for all $i \in Q$, we have $\pi(i) > 0$ 
and $\pi(i) = \sum_{j \in Q} \pi(j)\mathbf{P}_{ij}$). Furthermore,
 the expected return time $M_i$
to state $i$ satisfies $M_i = 1/\pi(i)$.



Let $(X_n)_{n \geq 0} = (X_0,X_1,X_2,\ldots)$ be the Markov
chain with transition matrix $\mathbf{P}$ and some initial distribution $\lambda$ on the states.
Consider, for each $i \in Q$, the stochastic variable
$V_i$, where
$$
V_i(n) = \sum_{k=0}^{n-1} 1_{X_k = i}, 
$$
\noindent that is, $V_i(n)$ is the number of times state $i$ is visited in the first $n$ elements of the Markov chain. 
As $\mathbf{P}$ is irreducible, the Ergodic Theorem
for  Markov chains (see, e.g., \cite[Thm.\ 75]{Serfozo:basics}) yields that, independently of $\lambda$, we have for arbitrary $\epsilon > 0$:
\begin{equation}\label{eq:YouSeeBigGirl}
\lim_{n\rightarrow\infty} \textrm{Pr}\left( \left\vert \frac{V_i(n)}{n} - \pi(i) \right\vert \geq \epsilon \right) =
\lim_{n\rightarrow\infty} \textrm{Pr}\left( \left\vert \frac{V_i(n)}{n} - \frac{1}{M_i} \right\vert \geq \epsilon \right) = 0
\end{equation}

Let $n$ be a positive integer, let $w = w_1 \cdots w_n \in \Sigma^n$,
and let $\mathbf{q}_{S_j}(w) = q^w_0 q^w_1 \cdots q^w_{n}$ be the
sequence of states visited in the run of $S_j$ on $w$  (i.e., $q^w_0 = j$). The probability of observing a state
sequence $q_0 \cdots q_n$ in the Markov chain is (when the initial distribution $\lambda$ has $\lambda(q_0) = \lambda(j) = 1$):
$$
\textrm{Pr}(q_0 \cdots q_n) = \prod_{i=0}^{n-1} \sum_{a \in \Sigma} p(a) 1_{\delta(q_i,a) = q_{i+1}} = \sum_{a_1,\ldots,a_n \in \Sigma} p(a_1)1_{\delta(q_0,a_1) = q_1} \cdots p(a_n)1_{\delta(q_{n-1},a_{n}) = q_n} 
$$
\noindent where we have used the fact that the Cauchy product of two absolutely convergent series  is convergent. 

As for all integers $i$ with  $0 \leq i \leq n$ we have $\delta(q^w_{i-1},w_i) = q^w_i$, we obtain:
\begin{equation*}
\sum_{a_1,\ldots,a_n \in \Sigma} p(a_1)1_{\delta(q_0,a_1) = q_1} \cdots p(a_n)1_{\delta(q_{n-1},a_{n}) = q_n} =
\mu_p(\{a_1 \cdots a_n : \mathbf{q}_{S_j}(a_1 \cdots a_n) = q_0 \cdots q_n\})
\end{equation*}
and hence
\begin{equation}
\textrm{Pr}(q_0 \cdots q_n) = \mu_p(\{w : \mathbf{q}_{S_j}(w) = q_0 \cdots q_n \})
\end{equation}
Thus, as $S$ is deterministic and every $w_1 \cdots w_n \in \Sigma^n$ occurs along exactly one path of states in $S$, we have:
\begin{align}
\textrm{Pr}\left( \left\vert \frac{V_i(n)}{n} - \pi(i)\right\vert \geq \epsilon\right) &= 
\sum_{q_0 q_1 \cdots q_n \in Q^n} \textrm{Pr}( q_0 \cdots q_{n}) 1_{\vert V_i(n)/n - \pi(i) \vert \geq \epsilon} \nonumber \\
&= \sum_{q_0 q_1 \cdots q_n \in Q^n}  \mu_p(\{w_1 \cdots w_n : \mathbf{q}_{S_j}(w_1 \cdots w_n) = q_0 \cdots q_n\}) \nonumber \\
&= \mu_p\left( w :  \left\vert \frac{V_i(n)}{n} - \pi(i)\right\vert \geq \epsilon  \right) \label{myHovercraftIsFullOfEels}
\end{align}

Hence, by Equations \ref{eq:YouSeeBigGirl} and \ref{myHovercraftIsFullOfEels}, we have
\begin{equation}\label{eq:limes_final}
\lim_{n\rightarrow \infty} \mu_p\left( w : \left\vert \frac{V_i(n)}{n} - \pi(i)\right\vert \geq \epsilon  \right) = 0
\end{equation}
If $\mathbf{q}_{S_j}(w) = q_0 \cdots q_{n}$ and $q_k \in Q_F$ for some $k$ with $0 \leq k \leq n-1$, then $S_j$ selects $w_{k+1}$. Set
$c = \min_{q_i \in Q_F} \pi(i)$ ($c$ is well-defined as $Q_F \neq \emptyset$), and let $i \in Q_F$ be such that $\pi(i) = c$. Then, for all
$j \in Q$:
\begin{align*}
 \mu_p(E_n(c-\epsilon,j) &=  \mu_p\left( \{w \in \Sigma^n : \vert S_j[w] \vert \leq (c - \epsilon) n \}\right)  
\leq \mu_p\left( \{w \in \Sigma^n : V_i(n) \leq (c - \epsilon) n \}\right) \\
&= \mu_p\left( \left\{w \in \Sigma^n : \frac{V_i(n)}{n} - c \leq - \epsilon \right\}\right) =  \mu_p\left( \left\{w \in \Sigma^n : \left\vert \frac{V_i(n)}{n} - c \right\vert \geq \epsilon \right\}\right)
\end{align*}
And hence, by Equation \ref{eq:limes_final}, we have $\lim_{n\rightarrow\infty} \mu_p(E_n(c-\epsilon,j)) = 0$, and as $j \in Q$ was arbitrary, we obtain

\begin{align*}
\lim_{n\rightarrow\infty} \mu_p(E_n(c-\epsilon)) &= \lim_{n\rightarrow\infty} \mu_p(\cup_{j \in Q}  \mu_p(E_n(c-\epsilon,j))  \leq
\lim_{n \rightarrow \infty} \sum_{j \in Q} \mu_p(E_n(c - \epsilon),j) \\
&= \sum_{j \in Q} \lim_{n\rightarrow\infty}  \mu_p(E_n(c - \epsilon),j) = 0
\end{align*}
\noindent as desired.
\end{proof}

\begin{lemma}\label{main:lemma2}
Let $S$ be a strategy, $a \in \Sigma$, $b, \epsilon$ be real numbers with
$0 < b \leq 1$ and $\epsilon > 0$, and $p : \Sigma \longrightarrow [0,1]$ be a positive Bernoulli distribution. Define, for all positive integers $n$:
\begin{align*}
H_{n}(b,\epsilon) &=\left\{\word{w}\in\Sigma^{n} : \vert S(\word{w})\vert >bn \land \left\vert p(a) - \frac{\countones{a}{S(\word{w})}}{\vert S(\word{w}) \vert} \right\vert \geq \epsilon\right\} \\
&= \bigcup_{bn < \ell \leq n} \left\{\word{w} \in \Sigma^n : S(\word{w}) \in \Sigma^\ell \land  \left\vert p(a) - \frac{\countones{a}{S(\word{w})}}{\ell} \right\vert \geq \epsilon \right\}
\end{align*}
Then:
$$
\lim_{n\rightarrow\infty}\mu_{p}(H_{n}(b,\epsilon))=0
$$
\end{lemma}

\begin{proof}
Define
$$
F_n(b,\epsilon) = \bigcup_{bn < \ell \leq n} \left\{ \word{y} \in \Sigma^{\ell} : \left\vert p(a) - \frac{\countones{a}{\word{y}}}{\ell} \right\vert \geq \epsilon \right\}
$$
Observe that 
$H_n(b,\epsilon) = \left\{\word{w} \in \Sigma^n : S(\word{w}) \in F_n(b,\epsilon) \right\}$. 
Thus, $\mu_p(H_n(b,\epsilon)) \leq \mu_p (F_n(b,\epsilon))$ for all $n$, and it thus suffices to prove that
$\lim_{n\rightarrow\infty} \mu_p(F_n(b,\epsilon)) = 0$.

Consider the stochastic variable $X_a$
that is $1$ when $a$ is picked from $\Sigma$
with probability $p(a)$, and $0$ otherwise. Then, the mean of $X_a$ is $p(a)$ and
the variance of $X_a$ is $p(a)(1-p(a))$. Now consider
performing $\ell \geq 1$ independent Bernoulli trials
drawn according to $X_a$. Define $q : \{0,1\}^+ \longrightarrow [0,1]$ inductively by $q(1) = p(a)$, $q(0) = 1-p(a)$, and
$q(1c) = p(a)q(c)$ and $q(0c) = (1-p(a))q(c)$
for $c \in \Sigma^+$, and observe that $q$ induces a probability 
distribution $\bar{q}$ on $\Sigma^\ell$ by setting $\bar{q}(w) = q(w)$. Now, for any 
$\word{v} \in \Sigma^\ell$, $\bar{q}(\word{v})$ is the probability of obtaining $\word{v}$ by performing $\ell$ independent Bernoulli trials as above.

Define the stochastic
variable $X^{\ell}_a = X_a + X_a + \cdots +  X_a$ ($\ell$ times). Then, $X^\ell$ counts the number of occurrences of $a$ by performing 
 the $\ell$ repeated Bernoulli trials.

By the Chernoff bound, $X^{\ell}_a$ satisfies:
\begin{equation}\label{main:eq:Chernoff}
\textrm{Pr}\left(\left\vert p(a) -  \frac{X^{\ell}_a}{\ell} \right\vert \geq \epsilon \right) \leq 2 e^{-\frac{\ell \epsilon^2}{3p(a)}}
\end{equation}

Define the map $g : \Sigma \longrightarrow \{0,1\}$ by 
$g(a) = 1$ and $g(b) = 0$ for all $b \in \Sigma \setminus \{a\}$. Clearly, $g$ 
extends homomorphically to a map $\tilde{g} : \Sigma^\ell \longrightarrow \{0,1\}^\ell$ by setting
$\tilde{g}(c_1 c_2 \cdots c_\ell) = g(c_1)g(c_2) \cdots g(c_\ell)$. 


\noindent \textbf{Claim:} For any
$\word{u} \in \{0,1\}^\ell$, 
\begin{equation}\label{eq:trivial-ish}
\bar{q}(\word{u}) = \mu_p(\{\word{y} \in \Sigma^\ell : \tilde{g}(\word{y}) = \word{u}\})
\end{equation}
\textbf{Proof of claim:}
By induction on $\ell$. 
\begin{itemize}
\item If $\ell = 1$, then if $\word{u} = 0$, we have
$\{\word{y} \in \Sigma^\ell : \tilde{g}(\word{y}) = \word{u}\} = \Sigma \setminus \{a\}$ and thus:
$$
\bar{q}(\word{u}) = \bar{q}(0) = q(0) =  1 - p(a) = \sum_{b \in \Sigma \setminus \{a\}} p(b) = \mu_p(\Sigma \setminus \{a\})
$$
Similarly, if $\word{u} = 1$, we have $\{\word{y} \in \Sigma^\ell : \tilde{g}(\word{y}) = \word{u}\} = \{a\}$,
and thus $\bar{q}(\word{u}) = \bar{q}(1) = q(1) = p(a) = \mu_p(\{a\})$, as desired.

\item If $\ell > 1$, write $\word{u} = b_1 \cdots b_{\ell - 1} b_\ell$; by the induction hypothesis:
$$
\bar{q}(b_1 \cdots b_{\ell - 1}) = \mu_p(\{\word{y} \in \Sigma^\ell : \tilde{g}(\word{y}') = b_1 \cdots b_{\ell - 1}\}) = \sum_{\substack{\word{y}' \in \Sigma^{\ell - 1} \\ \tilde{g}(\word{y}') = b_1 \cdots b_{\ell-1}}} \mu_p(\word{y}')
$$
If $b_\ell = 0$, then:
\begin{align*}
\bar{q}(b_1 \cdots b_{\ell - 1} b_\ell) &= \bar{q}(b_1 \cdots b_{\ell - 1})q(0) = \bar{q}(b_1 \cdots b_{\ell - 1})(1-p(a))
=  \sum_{\substack{\word{y}' \in \Sigma^{\ell - 1} \\ \tilde{g}(\word{y}')
 = b_1 \cdots b_{\ell-1}}} \mu_p(\word{y}') (1 - p(a)) \\ 
&= \sum_{\substack{\word{y}' \in \Sigma^{\ell - 1} \\ \tilde{g}(\word{y}') = b_1 \cdots b_{\ell-1}}} \left(\mu_p(\word{y}') \sum_{c \in \Sigma\setminus\{a\}} p(c) \right) = \sum_{\substack{\word{y}' \in \Sigma^{\ell - 1} \\ \tilde{g}(\word{y}') = b_1 \cdots b_{\ell-1}}} \sum_{c \in \Sigma\setminus\{a\}} \mu_p(\word{y}')p(c) \\
&\stackrel{(\dagger)}{=} \sum_{\substack{\word{y}' \in \Sigma^{\ell - 1} \\ \tilde{g}(\word{y}') = b_1 \cdots b_{\ell-1}\\ c \in \Sigma \setminus \{a\}}} \mu_p(\word{y}'c)
= \sum_{\substack{\word{y} \in \Sigma^{\ell} \\ \tilde{g}(\word{y}) = b_1 \cdots b_{\ell-1} b_\ell}} \mu_p(\word{y})\\
&= \mu_p(\{\word{y} \in \Sigma^\ell : \tilde{g}(\word{y}) = b_1 \cdots b_{\ell - 1} b_\ell\}) 
\end{align*}
\noindent where ($\dagger$) follows as both series on the left- and right-hand sides of the equality are absolutely convergent. 
The proof for the case $b_\ell = 1$ is symmetric, mutatis mutandis.
\end{itemize}
\noindent (End of proof of claim.)

Observe that, for any
$\word{y} \in \Sigma^\ell$, we have: 
\begin{equation}
\vert p(a) - \countones{1}{\word{\tilde{g}(\word{y})}}/\ell \vert \geq \epsilon \quad \textrm{ iff } \quad \vert p(a) - \countones{a}{\word{y}}/\ell \vert \geq \epsilon \label{eq:uneq_retract}
\end{equation}
Hence, by Equation (\ref{eq:trivial-ish}), for any event $\mathcal{U} \subseteq \{0,1\}^\ell$, we have: 
\begin{align}
\textrm{Pr}(\mathcal{U}) &= \sum_{\word{u} \in \mathcal{U}} \bar{q}(\word{u})=
\sum_{\word{u} \in \mathcal{U}} \mu_p(\{\word{y} \in \Sigma^\ell : \tilde{g}(\word{y}) = \word{u}\})) \nonumber\\
&= \mu_p\left( \left\{ \word{y} \in \Sigma^\ell :
\tilde{g}(\word{y}) \in \mathcal{U} \right\}\right)
\label{eq:move_to_y}
\end{align}

The event $\vert p(a) - X^{\ell}_a/\ell  \vert \geq \epsilon$
 is shorthand for the set
 \begin{align*}
  \left\{ \word{u} \in \{0,1\}^\ell : \left\vert p(a) - \frac{\sum_{j=1}^\ell u_j}{\ell} \right\vert \geq \epsilon\right\} &=
\left\{ \word{u} \in \{0,1\}^\ell : \left\vert p(a) - \frac{\countones{1}{\word{u}}}{\ell} \right\vert \geq \epsilon\right\} 
\end{align*}

We thus obtain:
\begin{align}
\textrm{Pr}\left( \left\vert p(a) - \frac{X^{\ell}_a}{\ell}  \right\vert \geq \epsilon \right) &= \textrm{Pr}\left( \left\{ \word{u} \in \{0,1\}^\ell : \left\vert p(a) - \frac{\countones{1}{\word{u}}}{\ell} \right\vert \geq \epsilon\right\} \right) \nonumber\\
&= \mu_p \left( \left\{ \word{y} \in \Sigma^\ell : \left\vert p(a) - \frac{\countones{1}{\word{\tilde{g}(\word{y})}}}{\ell}  \right\vert \geq \epsilon \right\}  \right) &\textrm{by } (\ref{eq:move_to_y}) \nonumber \\
&= \mu_p\left( \left\{ \word{y} \in \Sigma^\ell : \left\vert p(a) - \frac{\countones{a}{\word{y}}}{\ell}  \right\vert \geq \epsilon \right\} \right)  &\textrm{by } (\ref{eq:uneq_retract}) \label{it_foo}
\end{align}

Observe that: 
\begin{align*}
\mu_p\left(F_n(b,\epsilon) \right) &= \mu_p \left(\bigcup_{bn < \ell \leq n} \left\{ \word{y} \in \Sigma^{\ell} \cap F_n(b,\epsilon) : \left\vert p(a) - \frac{\countones{a}{\word{y}}}{\ell} \right\vert \geq \epsilon \right\} \right) \nonumber \\
&= \sum_{bn < \ell \leq n} \mu_p\left(  \left\{ \word{y} \in \Sigma^{\ell} \cap F_n(b,\epsilon) : \left\vert p(a) - \frac{\countones{a}{\word{y}}}{\ell} \right\vert \geq \epsilon \right\}  \right) \label{eq:stop_having_fun} \\
&\leq \sum_{bn < \ell \leq n}\mu_p\left( \left\{ \word{y} \in \Sigma^\ell : \left\vert p(a) - \frac{\countones{a}{\word{y}}}{\ell}  \right\vert \geq \epsilon \right\} \right) \\
&= \sum_{bn < \ell \leq n} \textrm{Pr}\left( \left\vert p(a) - \frac{X^{\ell}_a}{\ell}  \right\vert \geq \epsilon \right) &\textrm{by } \ref{it_foo}\\
&\leq  \sum_{bn < \ell \leq n}  2 e^{-\frac{\ell \epsilon^2}{3p(a)}} &\textrm{by } \ref{main:eq:Chernoff} \\
&\leq (1-b)n  2 e^{-\frac{bn \epsilon^2}{3p(a)}}
\end{align*}
And thus $\lim_{n\rightarrow\infty} \mu_p(F_n(b,\epsilon)) = 0$, as desired.

\end{proof}

\begin{corollary}\label{cor:G_tends}
Let $b,\epsilon$ be real numbers
with $0 < b \leq 1$ and $\epsilon > 0$.
Then,
$$
\lim_{n\rightarrow\infty}\mu_{p}(G_{n}(b,\epsilon))=0
$$
\end{corollary}

\begin{proof}
By Lemma \ref{main:lemma2} with
$S$ the strategy defined by the automaton $A_q$, we obtain that:
$$
\lim_{n\rightarrow\infty} \mu_p(G_n(b,\epsilon,q)) = 0
$$
and as $G_n(b,\epsilon) = \bigcup_{q\in Q}G_{n}(b,\epsilon,q)$,
we have:
$$
\mu_p(G_n(b,\epsilon)) \leq \sum_{q\in Q}\mu_p(G_n(b,\epsilon,q))
$$
As $Q$ is finite, we hence obtain
$\lim_{n \rightarrow\infty} \mu_p(G_n(b,\epsilon)) = 0$.
\end{proof}


\begin{lemma}\label{lem:mainclaim}
There is a real number $b$ with $0 < b \leq 1$ such that for all $\epsilon > 0$:
\[ \lim_{n\rightarrow\infty} \mu_{p}(D_{n}^{p}(b,\epsilon)) = 1. \]
\end{lemma}

\begin{proof}
Observe that, for all $b$ with $0 < b \leq 1$:
\begin{align*}
\Sigma^n \setminus D_{n}^{p}(b,\epsilon)) &= \left\{\word{w} \in \Sigma^n :
\exists q \in Q . \vert A_q[\word{w}] \vert \leq bn\right\}
\\ 
&\cup \left\{ \word{w} \in \Sigma^n :  \exists q \in Q . \vert
 A_q[\word{w}] \vert > bn \land \sup_{a \in \Sigma} \left\vert \frac{\countones{a}{A_q[\word{w}]}}{\vert A_q[\word{w}] \vert}  - p(a) \right\vert \geq \epsilon \right\} \\
 &= \left( \bigcup_{q \in Q} E_n(b,q) \right) \cup \left(\bigcup_{q \in Q} G_n(b,\epsilon,q) \right)
\end{align*}
and thus: 
\begin{align*}
\mu_p(\Sigma^n \setminus D_n^{p}(b,\epsilon)) &\leq
\mu_p\left(  \bigcup_{q \in Q} E_n(b,q)  \right) + \mu_p\left( \bigcup_{q \in Q} G_n(b,\epsilon,q)  \right)\\
&= \mu_p(G_n(b,\epsilon)) + \mu_p(E_n(b))
\end{align*}
By \autoref{main:lemma1}, choose a real number $c > 0$ such that
$\lim_{n\rightarrow\infty} \mu_{p}(E_{n}(c-\epsilon))=0$,
and set $b = c - \epsilon$.

By \autoref{cor:G_tends}, we obtain that
$\lim_{n\rightarrow \infty} G_n(b,\epsilon) = 0$,
and thus 
$\lim_{n\rightarrow\infty}\mu_p(\Sigma^n \setminus D_n^{p}(b,\epsilon)) = 0$.
The result now follows by $\mu_p(D_n^{p}(b,\epsilon))) = 1 - \mu_p(\Sigma^n \setminus D_n^{p}(b,\epsilon))$.
\end{proof}

\begin{lemma}\label{lem:Agafonov_preserve_arbitrary}
Let $p : \Sigma \longrightarrow [0,1]$ be a probability distribution, 
let $\alpha \in \Sigma^\omega$ be $\mu_p$-block-distributed, and $A$ a strongly connected DFA over $\Sigma$. 
Then,
 for all $a \in \Sigma$, the limiting frequency of $a$ in the sequence $\beta = A[\alpha]$ 
exists and equals $p(a)$.
\end{lemma}

\begin{proof}
For each $n,r$, let $\beta_{(n,r)}$ be the sequence of symbols picked out from
the block $\alpha_{(n,r)}$ when $A$ is applied to $\alpha$; note that
each $\beta_{(n,r)}$ has  length between $0$ and $n$. 

For each positive integer $m$, define:
$$
L_m =\sum_{i=1}^{m}\vert \beta_{(n,i)} \vert
$$ 
And for each $a \in \Sigma$, define $\rho_a^m$ by:
 $$
 \rho_a^m =\frac{\sum_{i=1}^{m} \countones{a}{\beta_{(n,i)}}}{L_m}
 $$ 

To prove the lemma, it suffices to show
 that, for any real number $\epsilon > 0$, then for all sufficiently large
 $m$, we have $\abs{\rho_a^m-p(a)}<\epsilon$.
 
Define:
$$
I_m=\left\{i\leq m : \alpha_{(n,i)}\not\in D_{n}^{p}\left(b,\frac{\epsilon}{2}\right)\right\}
$$
\noindent where  $b < 1$ is a constant to be fixed later in the proof.

And define:
$$
\ell_m =\sum_{i\in I_m}\vert \beta_{(n,i)} \vert
$$

Now, define $\theta_a^m$ by:
$$
\theta_a^m = \frac{\sum_{i \in \{1,\ldots,m\} \setminus I_m} \countones{a}{\beta_{(n,i)}}}{\sum_{i \in \{1,\ldots,m\} \setminus I_m} \vert \word{y}_{(n,i)}\vert} = \frac{\sum_{i \in \{1,\ldots,m\} \setminus I_m} \countones{a}{\beta_{(n,i)}}}{L_m - \ell_m}
$$
That is, $\theta_a^m$ is the frequency of occurrences of $a$ when the blocks $\beta_{(n,i)]}$ picked out from blocks $\alpha_{(i,r)} \in D_{n}^{p}(b,\frac{\epsilon}{2})$ with $i \leq m$ are all concatenated.
Observe that, by definition
of $D^p_n$, we have $\abs{\theta^m_a-p(a)}<\frac{\epsilon}{2}$.

We have:
 \begin{align}
 \rho_a^m - \theta_a^m &= 
 \frac{\sum_{i=1}^{m} \countones{a}{\beta_{(n,i)}}}{L_m} -
 \frac{\sum_{i\in \{1,\ldots,m\} \setminus I_m} \countones{a}{\beta_{(n,i)}}}{L_m - \ell_m} \nonumber\\
&=
 \left(\frac{\sum_{i\in I_m} \countones{a}{\beta_{(n,i)}}}{L_m} + \frac{\sum_{i \in \{1,\ldots,m\} \setminus I_m}
 \countones{a}{\beta_{(n,i)}}}{L_m}\right) - \frac{\sum_{i \in \{1,\ldots,m\} \setminus I_m}\countones{a}{\beta_{(n,i)}}}{L_m-\ell_m}
\nonumber \\
 &\overset{(\dagger)}{=}
 \frac{\sum_{i \in \{1,\ldots,m\} \setminus I_m}
 \countones{a}{\beta_{(n,i)}}}{L_m} - \frac{\sum_{i \in \{1,\ldots,m\} \setminus I_m}\countones{a}{\beta_{(n,i)}}}{L_m-\ell_m}
 +
  \frac{\sum_{i\in I_m} \countones{a}{\beta_{(n,i)}}}{L_m} \nonumber\\ 
  &\leq \frac{\sum_{i\in I_m} \countones{a}{\beta_{(n,i)}}}{L_m}
  \leq \frac{\sum_{i\in I_m} \vert \beta_{(n,i)}\vert}{L_m} = \frac{\ell_m}{L_m}
  \label{final_line_foo}
  \end{align}
\noindent where the penultimate inequalities in the last line above
follows because $L_m \geq L_m - \ell_m$ implies
$\frac{\sum_{i \in \{1,\ldots,m\} \setminus I}
 \countones{a}{\beta_{(n,i)}}}{L} - \frac{\sum_{i \in \{1,\ldots,m\} \setminus I}\countones{a}{\beta_{[n,i]}}}{L-\ell} \leq 0$,
 and the final inequality follows because
 $\sum_{i\in I_m} \countones{a}{\beta_{(n,i)}}\leq \sum_{i\in I}\vert \beta_{(n,i)} \vert = \ell_m$.
 
By basic algebra, we have: 
$$
\frac{\sum_{i \in \{1,\ldots,m\} \setminus I_m}
 \countones{a}{\beta_{(n,i)}}}{L_m} - \frac{\sum_{i \in \{1,\ldots,m\} \setminus I_m}\countones{a}{\beta_{(n,i)}}}{L_m-\ell_m}
 =
 \frac{-\ell_m \sum_{i\in\{1,\ldots,m\} \setminus I}\countones{a}{\beta_{(n,i)}}}{L_m(L_m-\ell_m)}
$$
and as
$$
\sum_{i\in\{1,\ldots,m\} \setminus I_m}\countones{a}{\beta_{(n,i)}} 
\leq \sum_{i\in\{1,\ldots,m\} \setminus I_m} \vert \beta_{(n,i)} \vert =
L_m-\ell_m
$$
we conclude that: 
$$
\frac{-\ell_m \sum_{i\in\{1,\ldots,m\} \setminus I_m}\countones{a}{\beta_{(n,i)}}}{L_m(L_m-\ell_m)} \geq -\frac{\ell_m}{L_m}
$$
and thus by ($\dagger$) above that:
$$
\rho_a^m - \theta_a^m - \frac{\sum_{i\in I_m} \countones{a}{\beta_{(n,i)}}}{L_m} \geq - \frac{\ell_m}{L_m}
$$
\noindent whence $-\ell_m/L_m \leq \rho^m_a - \theta^m_a$,
which combined with (\ref{final_line_foo}) yields
$\vert \rho^m_a - \theta^m_a \vert \leq \ell_m/L_m$.

By \autoref{lem:mainclaim} pick a $b$ such that such that for all $\epsilon > 0$, we have $\lim_{n\rightarrow\infty} \mu_{p}(D_{n}^{p}(b,\epsilon)) = 1$. Choose $\delta>0$ with $\delta<\frac{b\epsilon}{8}$, and pick $n\in\naturalN$ such that $\mu_{p}(D_{n}^{p}(b,\epsilon))>1-\delta$. Now, pick $\gamma<\frac{b\epsilon}{8}$. Because $\alpha$ is $\mu_p$-block-distributed, there exists $M \in\naturalN$ such that for all $k \geq M$ and
all $\seq{B}\subseteq \Sigma^{n}$, the prefix
$\alpha \vert_{\leq kn}$ satisfies:
$$
\left\vert 
\frac{\vert \{ i \leq k : \alpha_{(n,i)} \in \seq{B}\} \vert}{k} - \mu_p(\seq{B})
\right\vert < \gamma
$$
In the particular case $\word{B} = D^p_n(b,\epsilon/2)$, we thus have:
$$
\left\vert \frac{\vert \{ i \leq k : \alpha_{(n,i)} \in D_n^p\left(b,\frac{\epsilon}{2}\right)\} \vert}{k} - \mu_p\left( D^p_n\left(b,\frac{\epsilon}{2}\right) \right) 
\right\vert < \gamma
$$
and thus
{\small
\begin{align*}
1 - \delta  - \frac{\vert \{ i \leq k : \alpha_{(n,i)} \in D_n^p(b, \frac{\epsilon}{2})\} \vert}{k} &\leq 
\mu_p\left( D^p_n\left( b,\frac{\epsilon}{2} \right) \right)
- \frac{\vert \{ i \leq k : \alpha_{(n,i)} \in D_n^p(b,\frac{\epsilon}{2})\} \vert}{k} 
 < \gamma
\end{align*}
}
whence we conclude:
\begin{equation}\label{eq:almost_final}
\left\vert \left\{ i \leq k : \alpha_{(n,i)} \in D_n^p\left( b,\frac{\epsilon}{2}\right)\right\} \right\vert > k(1-\delta - \gamma)
\end{equation}

By definition of $D_{n}^{p}(b,\frac{\epsilon}{2}))$, every
$\alpha_{(n,i)} \in D_{n}^{p}(b,\frac{\epsilon}{2}))$
satisfies $\vert \pickedout{\alpha_{(n,i)}}{A} \vert > bn$,
and we thus have:
\begin{equation}\label{eq:big_l}
L_m = \sum_{i=1}^m \vert \word{y}_{(n,i)} \vert
= \sum_{i=1}^m \vert  \pickedout{\alpha_{(n,i)}}{A} \vert
\geq \left\vert \left\{ i \leq m : \alpha_{(n,i)} \in D_n^p\left( b,\frac{\epsilon}{2}\right)\right\} \right\vert bn > m(1 - \delta - \gamma)bn
\end{equation}
Furthermore, by the definition of $I_m$ and (\ref{eq:almost_final}):
\begin{align*}
\vert I_m \vert &= \left\vert \left\{i\leq m : \alpha_{(n,i)}\not\in D_{n}^{p}\left(b,\frac{\epsilon}{2}\right)\right\}\right\vert = m - \left\vert \left\{ i \leq m : \alpha_{(n,i)} \in D_n^p\left( b,\frac{\epsilon}{2}\right)\right\} \right\vert\\
&< m - m(1-\delta-\gamma) = m(\delta+\gamma)
\end{align*}
But then,
\begin{equation}\label{eq:small_l}
\ell_m = \sum_{i \in I_m} \vert \word{y}_{(i,n)} \vert \leq \vert I_m \vert n < mn(\delta + \gamma)
\end{equation}
\noindent and thus by \ref{eq:big_l} and \ref{eq:small_l}:
$$
\frac{\ell_m}{L_m} < \frac{mn(\delta + \gamma)}{m(1 - \delta -\gamma)bn} = \frac{\delta + \gamma}{b(1 - \delta - \gamma)} < \frac{\frac{b\epsilon}{8} + \frac{b\epsilon}{8}}{b\left( 1- \frac{b\epsilon}{8} - \frac{b\epsilon}{8}\right)}
< \frac{\frac{\epsilon}{4}}{1 - \frac{1}{4}} < \frac{\epsilon}{2}
$$
where we have used that $b\epsilon < 1$ in the penultimate inequality.

We now finally have
$$
\vert \rho_a -p(a) \vert \leq \vert\rho^m_a - \theta^m_a \vert + \vert \theta_a -p(a) \vert < \frac{\ell_m}{L_m} + \frac{\epsilon}{2}
< \frac{\epsilon}{2} + \frac{\epsilon}{2} = \epsilon
$$
\noindent concluding the proof.
\end{proof}

\begin{lemma}\label{lem:final_crucial}
Let $\Sigma$ be an alphabet, $p$ a positive Bernoulli distribution on $\Sigma$, let $\alpha \in \Sigma^\omega$ be $\mu_p$-distributed, and let $A$ be a strongly connected DFA over $\Sigma$.
 Then, $A[\alpha]$ is $\mu_p$-distributed.
\end{lemma}

\begin{proof}
By Lemma \ref{lem:simply_normal_is_enough} it suffices to show for every $a \in \Sigma$ and every strongly connected $A$ that the limiting frequency of $a$ in $A[\alpha]$ exists and equals $p(a)$. As $\alpha$ is $\mu_p$-distributed, it follows from Proposition \ref{prop:mu_implies_mu} that it is $\mu_p$-block-distributed, and the result then immediately follows by Lemma \ref{lem:Agafonov_preserve_arbitrary}.
\end{proof}

\section[An application in symbolic dynamics]{An application in symbolic dynamics: characterizing measures where genericity is preserved by DFAs}\label{sec:app_dyn}

We now show an application of the main result to the area of symbolic dynamical systems.
The following section recalls basic facts about symbolic dynamical systems, including establishing the correspondence between probability maps on
$\Sigma^*$ and probability measures on full shifts.

\subsection{Shift spaces and genericity}

We briefly introduce basic notions; full accounts can be found in standard textbooks, e.g.\ \citep{lind_marcus_1995}.

\begin{definition}
Let $\Sigma$ be a non-empty alphabet. The (one-sided) \emph{shift} $s : \Sigma^\omega \longrightarrow \Sigma^\omega$ is the map defined by
$s(a_1 a_2 a_3 \cdots) = a_2 a_3 \cdots$. A \emph{shift space} is a pair $(X,s)$ where $X \subseteq \Sigma^\omega$ is a closed
(in the product topology on $\Sigma^\omega$ when $\Sigma$ is endowed with the discrete topology) subset such that
$s(X) = X$\footnote{For one-sided shifts, some authors require only $s(X) \subseteq X$; we shall not do so here.}, and $s$ is the restriction of the shift to $X$.
\end{definition}

As usual, we consider the $\sigma$-algebra $\mathcal{C}$ on $\Sigma^\omega$ having the set of cylinders $\{\cylinder{\word{w}} : \word{w} \in \Sigma^*\}$ as basis. All measures $\mu$ in the remainder of the paper are understood to be measures on $(\Sigma^\omega, \mathcal{C})$.

The standard example of probability measures on shift spaces is the set of Bernoulli measures \citep{shields:bernoulli}:

\begin{definition}
A probability measure on the shift space $(\Sigma^\omega,s)$ is a probability measure on $\Sigma^\omega$ with the $\sigma$-algebra generated by the cylinder sets $\{\cylinder{\word{v}} : \word{v} \in \Sigma^*\}$.
A probability measure $\bar{\mu}$ on the full shift is a \emph{Bernoulli measure} if there is a probability distribution $p : \Sigma \longrightarrow [0,1]$ such that
the measure of each cylinder satisfies $\bar{\mu}(\cylinder{\word{a_1 \cdots a_n}}) = \prod_{i=1}^n p(a_i)$. In this case, we say that $\bar{\mu}$ is \emph{induced} by $p$.
\end{definition}

\begin{definition}
Let $(X,s)$ be a shift space. A probability measure $\bar{\mu}$ on $X$ is said to be \emph{shift invariant} if $\bar{\mu}(s^{-1}(A)) = \bar{\mu}(A)$ for all
$A \subseteq X$. A finite word $\word{w} \in \Sigma^k$ is said to be \emph{admissible} for $\mu$ if 
$\bar{\mu}(\cylinder{\word{w}}) > 0$.

A right-infinite sequence $\alpha \in \Sigma^\omega$ is said to be \emph{generic} for $\bar{\mu}$ if,
for all words $\word{w}$ admissible for $\bar{\mu}$, we have: 
$$
\lim_{n\rightarrow \infty} \frac{\countones{\word{w}}{\alpha\vert_{\leq n}}}{n} = \bar{\mu}(\cylinder{\word{w}})
$$
That is, $\word{w}$ occurs in $\alpha$ with limiting frequency $\bar{\mu}(\cylinder{\word{w}})$.
\end{definition}

The study of probability measures on the full shift is cryptomorphic to the study of invariant probability maps; this folklore result is contained in the following two propositions (proofs can be found in Appendix \ref{sec:auxiliary}).

\begin{proposition}\label{prop:back_and_forth_induce}
Every invariant probability map $\mu: \Sigma^* \longrightarrow  [0,1]$ induces a shift-invariant probability measure
$\bar{\mu} : \Sigma^\omega \longrightarrow [0,1]$ by setting $\bar{\mu}(\cylinder{\word{w}}) = \mu(\word{w})$.
Conversely, every probability measure $\nu :\Sigma^\omega \longrightarrow [0,1]$ induces a
probability map $\underline{\nu} : \Sigma^* \longrightarrow [0,1]$ by defining
$\underline{\nu}(\word{w}) = \nu(\cylinder{\word{w}})$; if $\nu$ is shift-invariant, then $\underline{\nu}$ is invariant.
Furthermore, $\mu = \underline{\bar{\mu}}$,
and $\nu = \bar{\underline{\nu}}$.
\end{proposition}

\begin{proposition}\label{prop:back_and_forth_final}
Let $\mu : \Sigma^* \longrightarrow [0,1]$ be a probability map. The following are equivalent:

\begin{enumerate}

\item \label{it:1_1} There exists a $\mu$-distributed  $\alpha \in \Sigma^\omega$.

\item \label{it:1_2} $\mu$ is invariant.

\item \label{it:1_3}  There exists a shift-invariant probability measure $\nu$ on $\Sigma^\omega$ such that $\bar{\mu} = \nu$.

\end{enumerate}
Conversely, let $\nu$ be a probability measure on $\Sigma$. The following are equivalent:

\begin{enumerate}

\item \label{it:2_1} There exists $\alpha \in \Sigma^\omega$ that is generic for $\nu$.

\item \label{it:2_2} $\nu$ is shift-invariant.

\item \label{it:2_3} There exists an invariant probability map $\mu : \Sigma^* \longrightarrow [0,1]$ such that $\underline{\nu} = \mu$.

\end{enumerate}
\end{proposition}

It follows that the shift-invariant  probability measures $\nu$ on the full shift such that genericity is preserved by finite-state selection,
are exactly the Bernoulli measures:

\begin{theorem}
Let $\Sigma$ be a non-empty alphabet, and let $\nu$ be a shift-invariant measure on the full shift $(\Sigma^\omega,s)$ such that there exists at least one
$\alpha \in \Sigma^\omega$ generic for $\nu$. Then,
every  finite-state selector preserves genericity if{f} $\nu$ is a Bernoulli measure such that all words in $\Sigma^*$ are admissible.
\end{theorem}

\begin{proof}
Observe that for a Bernoulli measure $\bar{\mu}$ on the full shift on $\Sigma$, all words are admissible if{f} $\bar{\mu}(a) > 0$ for all $a \in \Sigma$. The Theorem now follows
from Theorem \ref{the:main} and Proposition \ref{prop:back_and_forth_final}.
\end{proof}

\section{Future work}

The most obvious extension of our main results is to attempt to relax the requirement that selection is done by a DFA by using methods similar to
Kamae and Weiss \citep{KamaeWeiss}, and Kamae and Wang \citep{KamaeWang} where reasoning using a combination of density arguments and
relaxed finiteness conditions on the syntactic monoid of the strategy (using our terminology) have been used for normal sequences
over binary alphabets. We conjecture that some of these techniques can be adapted to positive Bernoulli distributions on arbitrary finite alphabets. 


A different possible thrust is to consider generalizations of Agafonov's Theorem on domains different from infinite sequence over alphabets. However, some
results in the -- sparse -- literature on selection from normal sequence-like objects in other contexts are negative; for example
normality is not preserved by arithmetic progressions (so, probably not by finite-state selectors in any reasonable sense) for continued fraction expansions \citep{Heersinkarticle}. On the other hand, very recent work by Bergelson et al.\ has succesfully adapted the classical
techniques of Kamae and Weiss \citep{KamaeWeiss} to show that certain F{\o}lner sequences preserve (the appropriate analogue of) normality
in cancellative amenable semigroups \citep{bergelson2020deterministic}.

\clearpage


\bibliographystyle{msclike}
\bibliography{../agafonovbib.bib}   

\clearpage

\appendix

\section{Auxiliary proofs and definitions}\label{sec:auxiliary}

\subsection{Automata and selectors}

The following is a proof of the extension of Lemma 2.6 of \citep{DBLP:journals/acta/SchnorrS72}. The proof follows the original in most details.

\begin{definition}
Let $G = (V,E)$ be a directed multigraph, and denote by $\sim \subseteq V \times V$ the equivalence
relation such that $v \sim w$ if{f} there is a directed path from $v$ to $w$ and a directed path from $w$ to $v$. 
For every $v \in V$, denote by $[v]_\sim$ the equivalence class containing $v$.
Define the partial order $<$ on $V/\!\!\sim$ by 
$\mathcal{V} < \mathcal{W}$ \textrm{ if{f} } there are $v \in \mathcal{V}$ and  $w \in \mathcal{W}$ such that there is a directed path from $w$ to $v'$.
\end{definition}

If $G$ has a finite number of nodes, $<$ is clearly well-founded. As $<$ is clearly also transitive,
every $\mathcal{W} \in  V/\!\!\sim$ satisfies $\mathcal{W} > \mathcal{V}$ for some $<$-minimal $\mathcal{V} \in V/\!\!\sim$.

Also observe that every $<$-minimal $\mathcal{V}$ is a recurrent strongly connected component, because (i) it is strongly connected by definition,
and $<$-minimality implies that no directed path from any node in $\mathcal{V}$ can reach a node in a strongly connected component distinct
from $\mathcal{V}$.

\begin{lemma}\label{lem:aux_ss}
Let $S = (Q,\delta,q_s,Q_F)$ be a finite automaton over a (possibly infinite) alphabet $\Sigma$. Then there is a word $w \in \Sigma^*$
such that, for all states $q \in Q$, $\delta^*(q,w)$ is a state in a $<$-minimal element of $Q/\!\!\sim$.
\end{lemma}

\begin{proof}
Write $Q = \{q_1,\ldots,q_m\}$. We prove by induction on $i \leq m$ that there is a word $w_i \in \Sigma^*$ such that 
for all $j \leq i$, $\delta^*(s_j,w_i)$ is a state in a $<$-minimal element of $Q/\!\!\sim$.
\begin{description}

\item[$i=1$:] Let $\mathcal{V}$ be a $<$-minimal element of $Q/\!\!\sim$ such that $[q_1]_\sim > \mathcal{V}$. Choose $q \in Q$
such that $[q]_\sim = \mathcal{V}$. Then there is a directed path from $s_1$ to $q$. Let $w_1$ be the word along that path,
and observe that $\delta^*(q_1,w_1) = q$.

\item[$i>1$:]  Let $\mathcal{V}$ be a $<$-minimal element of $Q/\!\!\sim$ such that $\delta^*(q_{i+1},w_i) \in \mathcal{V}$,
and let $q \in \mathcal{V}$, whence there is a directed path from $\delta^*(q_{i+1},w_i)$ to $q$. Let $w' \in \Sigma^*$ be the word
along that path, whence $\delta^*(\delta^*(q_{i+1},w_i),w') = q$. Define $w_{i+1} = w_i \cdot w'$, and observe that
$\delta^*(q_{i+1},w_{i+1}) = q$. 

For $j \leq i$, we claim that $\delta^*(q_{j},w_{i+1})$ is a state in a $<$-minimal element of $Q/\!\!\sim$. For, by the Induction Hypothesis,
$\delta(q_j,w_i)$ is in a $<$-minimal element $\mathcal{V}_j$ of $Q/\!\!\sim$, and as  $<$-minimal element are recurrent strongly connected components,
no directed path from $\delta(q_j,w_i)$ can end in a state outside $\mathcal{V}_j$.

\end{description}
\end{proof}

\begin{proof}[Proof of Lemma \ref{lem:SS}]
Any $<$-minimal element of  $Q/\!\!\sim$ is recurrent. By Lemma \ref{lem:aux_ss}, there is a word $w$ such that
from any state $q \in Q$, $\delta^*(q,w)$ is a state in a recurrent strongly connected component of the automaton. As $p(a) > 0$ for all $a \in \Sigma$, $\mu_p(w) > 0$, and as $\alpha$ is $p$-distributed,
$w$ thus occurs (infinitely often) in $\alpha$. After the first occurrence of $w$, the run of $A$ on $\alpha$ has entered a strongly recurrent connected component.
\end{proof}

\subsection{$\mu$-distribution}

\begin{proof}[Proof of Proposition \ref{prop:mu_implies_mu}]
We use exactly the same arguments as in the proof by Niven and Zuckerman \citep{niven1951}, but using the notation of the present paper. Almost the entirety of the proof in \citep{niven1951} is devoted
to counting arguments on finite prefixes of $\alpha$, and involves neither the size of the alphabet $\Sigma$, nor the particular distribution on it; indeed
any consideration of those matters is isolated to a few observations in the beginning of the proof that are then used repeatedly when taking limits
later on. We have clearly indicated those observations below, but give the entirety of the proof in the interest of completeness.

Let $w = w_1 \cdots w_v \in \Sigma^v$ be arbitrary. We introduce the following notation:

\begin{itemize}
\item For any $t \geq 0$,  $w\Sigma^tw$ is the set
$\{w u w : u \in \Sigma^t\}$.

\item $\#_w^i(n)$ is the number of times that $w$ occurs in $\alpha\vert_{\leq n}$ at a position
congruent to $i$ (mod $n$).

\item $\#^{i,j}_w(n) = \#_w^{i}(n) - \#_w^{j}(n)$.

\item $g: \mathbb{N} \longrightarrow \mathbb{N}$ is the function defined by: $g(n) = \sum_{i=1}^{n-1} \#_w^i(n)$.

\item $\theta_t(n)$ is the number of occurrences of any element from $w \Sigma^t w$ in $\alpha\vert_{\leq n}$.

\item $w'$ is shorthand for any string of length between $v+1$ and $2v-1$ whose first $v$ digits are $w$ and whose last digits 
are $w$, i.e. an ``overlap of $w$ with itself''. Such a string does not necessarily exist.

\end{itemize}

We now treat the part of the proof depending on the cardinality of $\Sigma$ and $\mu_p$-distributedness (as opposed to finiteness of $\Sigma$ and equidistribution ).

As $\alpha$ is $\mu_p$-distributed, we have 
\begin{equation}\label{eq:ordinary_mu}
\lim_{n \rightarrow \infty} \frac{g(n)}{n} = \mu_p(w)
\end{equation}
and for each fixed $t \geq 0$, we also have:
\begin{align}
\lim_{n \rightarrow \infty} \frac{\theta_t(n)}{n} &= 
\lim_{n \rightarrow \infty}  \frac{\sum_{a_1 \cdots a_t \in \Sigma^t}\#_{wa_1 \cdots a_t w}(\alpha \vert_{\leq n})}{n} & \nonumber \\
&= \lim_{n \rightarrow \infty}  \sum_{a_1 \cdots a_t \in \Sigma^t} \frac{\#_{wa_1 \cdots a_t w}(\alpha \vert_{\leq n})}{n} & \nonumber\\
&=  \sum_{a_1 \cdots a_t \in \Sigma^t} \lim_{n \rightarrow \infty}  \frac{\#_{wa_1 \cdots a_t w}(\alpha \vert_{\leq n})}{n} & \textrm{(By the Dominated Convergence Theorem)} \nonumber \\
&= \sum_{a_1 \cdots a_t \in \Sigma^t} \mu_p(w a_1 \cdots a_t w) & \nonumber \\
&= \sum_{a_1 \cdots a_t \in \Sigma^t} \mu_p(w)^2 \prod_{i=1}^t p(a_i) & (\textrm{As $\mu_p$ is induced by a Bernoulli distribution}) \nonumber\\
&=   \mu_p(w)^2 \sum_{a_1 \cdots a_t \in \Sigma^t} \prod_{i=1}^t p(a_i) & \nonumber \\
&= \mu_p(w)^2 \prod_{i=1}^t \sum_{a \in \Sigma} p(a) & \textrm{(By monotone convergence)} \nonumber \\
&= \mu_p(w)^2 & \textrm{(As $1 = \Sigma_{a \in \Sigma} p(a)$)} \label{it:theta_is_good}
\end{align}

We shall prove that:
\begin{equation}\label{eq:congruence_close}
\lim_{n \rightarrow \infty} \frac{\#_w^{i,j}(n)}{n} = 0
\end{equation}
By \ref{eq:ordinary_mu} and \ref{eq:congruence_close},  it follows for any $i$ with $0 \leq i < v$ that:
$$
\lim_{n\rightarrow \infty}\frac{\#_w^{i}(n)}{n} = \frac{\mu_p(w)}{v}
$$
\noindent and as $v$ and $w \in \Sigma^v$ were arbitrary, that $\alpha$ is $\mu_p$-block-distributed.

The remainder of the proof is devoted to prove \ref{eq:congruence_close} and is only concerned with counting arguments on finite prefixes of $\alpha$. All arguments from hereon are, modulo notation and use of \ref{it:theta_is_good},
completely identical to the proof in \citep{niven1951}.

Let $s \geq 0$ be an integer. Observe that $\#_w^i(n+s) - \#^i_w(n)$ is the number of occurrences of $w$ that (1) are in $\alpha \vert_{\leq n + s}$ at a position congruent to $i \textrm{ mod } v$,
but (2) are not entirely contained in $\alpha\vert_{\leq n}$. Thus,
$$
\sum_{\substack{i < j\\ i \in \{0,\ldots,v-2\} \\ j \in \{1,\ldots,v-1\}}} \left( \#_w^i(n + s) - \#_w^i(n)\right)\left(\#_w^j(n+s) - \#_w^j(n) \right)
$$
is the number of words on the form $w'$ or $wuw$ (for  $0 \leq \vert u \vert \leq s- v - 1$ and $\neg(\vert u \vert \equiv 0 \textrm{ mod } v)$)
that occur in $\alpha\vert_{n+s}$, but such that the initial $\vert w \vert = v$ symbols are not entirely contained in $\alpha\vert_{\leq n}$.  

For $n > s$, we define:
\begin{equation}\label{eq:defining_sigma}
\sigma = \sum_{m=0}^{s-s} \sum_{\substack{i < j\\ i \in \{0,\ldots,v-2\} \\ j \in \{1,\ldots,v-1\}}} \left( \#_w^i(m + s) - \#_w^i(m)\right)\left(\#_w^j(m+s) - \#_w^j(m) \right)
\end{equation}

Consider $\alpha\vert_{\leq n}$ and any single occurrence of an element $wuw$ with $t = \vert u \vert \leq s - v - 1$  in $\alpha\vert_{\leq n}$. The occurrence of $wuw$ is counted a number of times in $\sigma$, and it it does not occur t too close to either end of $\alpha\vert_{\leq n}$,
it is counted at most $s - t - v$ times in $\sigma$. If the occurrence of $wuw$ is preceded by at least $s - t - 2v$ symbols and is followed by at least $s - t- v - 1$ symbols, $wuw$ is counted exactly $s - t - v$ times. Thus, we have:

\begin{equation}\label{eq:sigma_upper}
\sigma \geq \sum_{\substack{t=0\\ \neg(t \equiv 0 (\textrm{ mod } v))}}^{s-v-1} (s - t - v)(\theta_t(n-s) - \theta_t(s))
\end{equation}
Thus, by \ref{it:theta_is_good} and \ref{eq:sigma_upper}, we obtain, for \emph{fixed} $s$:
\begin{align}
\lim_{n \rightarrow \infty} \frac{\sigma}{n} &\geq  \lim_{n \rightarrow \infty} \left( \sum_{\substack{t=0\\ \neg(t \equiv 0 (\textrm{ mod } v))}}^{s-v-1} (s - t - v)(\theta_t(n-s) - \theta_t(s)) \right) & \nonumber \\
&= \sum_{\substack{t=0\\ \neg(t \equiv 0 \, (\textrm{mod } v))}}^{s-v-1} \lim_{n \rightarrow \infty} \left( \frac{\sigma}{n}  (s - t - v)(\theta_t(n-s) - \theta_t(s)) \right) & \nonumber \\
&= \sum_{\substack{t=0\\ \neg(t \equiv 0 \, (\textrm{mod } v))}}^{s-v-1} (s - t - v)\mu_p(w)^2 & \textrm{(By \ref{it:theta_is_good})} \label{it:close_to_end}
\end{align}

Choose $s$ such that $s \equiv 0 \, (\textrm{mod } v)$. Then, \ref{it:close_to_end} becomes:
\begin{equation}\label{eq:get_Brexit_done!}
\lim_{n \rightarrow \infty} \frac{\sigma}{n} \geq \frac{(v-1)(s-v)^2}{2v} \mu_p(w)^2
\end{equation}

Similarly, we count the number of occurrences of words on the form $wuw$ where $t = \vert u \vert \equiv 0 \, (\textrm{ mod } v)$ and proceed as above. This yields:
\begin{align}
\lim_{n \rightarrow \infty} \frac{1}{n}\sum_{m=0}^{n-s} \sum_{i=0}^{v-1} \frac{(\#_w^i(m+s) - \#_w^i(m))(\#_w^i(m+s) - \#_w^i(m) - 1)}{2}\\
= \sum_{\substack{t=0\\ \neg(t \equiv 0 (\textrm{ mod } v))}} (s - t - v) \mu_p(w)^2 = \frac{s(s-v)}{2v} \mu_p(w)^2
\end{align}

We have, for fixed $s$ such that $s \equiv 0 \, (\textrm{ mod } v)$, that:
\begin{align}
\lim_{n \rightarrow \infty} \frac{1}{2n} \sum_{m=0}^{n-s} \sum_{i=0}^{v-1} (\#_w^i(m+s) - \#_w^i(m)) &=
\lim_{n \rightarrow \infty} \frac{1}{2n} \sum_{m=0}^{n-s}(g(m+s) - g(m)) \nonumber \\
&= \lim_{n \rightarrow \infty} \left( \frac{1}{2n} \sum_{m=n-s+1}^n g(m) - \frac{1}{2n}\sum_{m=0}^{s-1} g(m) \right) \nonumber \\
&= \lim_{n \rightarrow \infty} \frac{1}{2n} \sum_{m=n-s+1}^n g(m) \nonumber \\
&= \frac{s \mu_p(w)}{2} \quad \quad \quad \textrm{(By \ref{eq:ordinary_mu})} \label{eq:get_the_damn_thing_done}
\end{align}

Thus, by \ref{eq:get_Brexit_done!} and \ref{eq:get_the_damn_thing_done}, we obtain:
\begin{equation}\label{eq:sheesh_when_are_we_done?}
\lim_{n\rightarrow\infty} \frac{1}{n} \sum_{m=0}^{n-s}\sum_{i=0}^{v-1} (\#_w^i(m+s) - \#_w(m))^2 = s\mu_p(w) + \frac{s(s-v)}{v} \mu_p(w)^2
\end{equation}

For fixed $s$ with $s \equiv 0 \, \textrm{ mod } v)$, \ref{eq:defining_sigma}, \ref{eq:get_Brexit_done!}, and \ref{eq:sheesh_when_are_we_done?} now yield:

\begin{align}
&\lim_{n \rightarrow \infty} \frac{1}{n} \sum_{m=0}^{n-s} \sum_{\substack{i < j\\ i \in \{0,\ldots,v-2\} \\ j \in \{1,\ldots,v-1\}}} \left( \#_w^i(m+s) - \#_w^i(m) - (\#_w^j(m+s) - \#_w^j(m)) \right) \nonumber \\
&\leq (v-1)s \mu_p(w) + (v-q)(s-v) \mu_p(w)^2 \label{eq:nobody_knows_the_trouble_Ive_seen}
\end{align}

Noting that $\sum_{i=1}^n x_i^2 \geq \frac{1}{n}(\sum_{i=1}^n x_i)^2$, we obtain:

\begin{align}
&\sum_{m=0}^{n-s}  \left( \#_w^i(m+s) - \#_w^i(m) - (\#_w^j(m+s) - \#_w^j(m)) \right) \nonumber\\
&\geq \frac{1}{n-s+1}\left( \sum_{m=0}^{n-s} \#_w^i(m+s) - \#_w^i(m) - \#_w^j(m+s) + \#_w^j(m)  \right)^2 \nonumber \\
&= \frac{1}{n-s+1} \left( \sum_{m=0}^{n-s} (\#_{w}^{i,j}(m+s) - \#_w^{i,j}(m))\right)^2 \nonumber \\
&= \frac{1}{n-s+1} \left( \sum_{m=0}^{s-1} \#_w^{i,j}(n-m) - \sum_{m=0}^{s-1}\#_w^{i,j}(m)\right)^2 \label{eq:please_release_me!}
\end{align}

Now, \ref{eq:nobody_knows_the_trouble_Ive_seen} and \ref{eq:please_release_me!} imply:

\begin{align}
&\limsup_{n\rightarrow\infty} \frac{1}{n(n-s+1)} \sum_{\substack{i < j\\ i \in \{0,\ldots,v-2\} \\ j \in \{1,\ldots,v-1\}}}  \left( \sum_{m=0}^{s-1} \#_w^{i,j}(n-m) - \sum_{m=0}^{s-1}\#_w^{i,j}(m)\right)^2\\
&\leq (v-1)s \mu_p(w) + (v-1)(s-v)\mu_p(w)^2 \label{eq:please_please_no_more}
\end{align}

By the definition of $\#_w^{i,j}$, we have $\vert \#_w^{i,j}(m) \vert < m$, and thus, for fixed $s$, we have:
$$
\lim_{n \rightarrow \infty} \frac{1}{n(n-s+1)} \left(\sum_{m=0}^{s-1} \#_w^{i,j}(m) \right)^2 = 0
$$
\noindent
and:
$$
\lim_{n \rightarrow \infty} \frac{1}{n(n-s+1)} \sum_{m=0}^{s-1} \#_w^{i,j}(n-m) \sum_{m=0}^{s-1} \#_w^{i,j}(m) = 0
$$
which in turn imply, by \ref{eq:please_please_no_more}, that:
\begin{align}
&\limsup_{n\rightarrow \infty} \frac{1}{n(n-s+1)} \sum_{\substack{i < j\\ i \in \{0,\ldots,v-2\} \\ j \in \{1,\ldots,v-1\}}}
\left( \sum_{m=0}^{s-1} \#_w^{i,j} (n-m) \right)^2 \nonumber \\
&= \limsup_{n\rightarrow \infty} \frac{1}{n(n-s+1)} \sum_{\substack{i < j\\ i \in \{0,\ldots,v-2\} \\ j \in \{1,\ldots,v-1\}}}
\left( s\#_w^{i,j}(n) + \sum_{m=0}^{s-1} (\#_w^{i,j}(n-m) - \#_w^{i,j}(n))\right)^2 \nonumber \\
&\leq (v-1)s \mu_p(w) + (v-1)(s-v)\mu_p(w)^2 \label{eq:make_it_stop!}
\end{align}

But $\vert \#_w^{i,j}(n-m) - \#_w^{i,j}(n)\vert < 2m$, so from \ref{eq:make_it_stop!} we obtain:
\begin{align}
&\limsup_{n \rightarrow \infty} \frac{1}{n(n-s+1)}  \sum_{\substack{i < j\\ i \in \{0,\ldots,v-2\} \\ j \in \{1,\ldots,v-1\}}} s^2 (\#_w^{i,j}(n))^2 
\leq (v-1)s \mu_p(w) + (v-1)(s-v)\mu_p(w)^2
\end{align}

which in turn implies, for fixed $s \equiv 0 \, (\textrm{ mod } v)$, that:
\begin{equation}
\limsup_{n \rightarrow \infty} \frac{(\#_w^{i,j}(n))^2}{n^2} = \limsup_{n \rightarrow \infty} \frac{(\#_w^{i,j}(n))^2}{n(n-s+1)}
\leq \frac{(v-1)\mu_p(w)}{s} + \frac{(v-1)(s-v)\mu_p(w)^2}{s^2} \label{eq:finally!}
\end{equation}
As the expression on the right-hand side of \ref{eq:finally!} can be made arbitrarily small by choosing $s$ large enough, we obtain:
$$
\lim_{n \rightarrow \infty} \frac{\vert \#_w^{i,j}(n) \vert}{n} = 0
$$
and hence:
$$
\lim_{n \rightarrow \infty}\frac{\#_w^i(n)}{n} = \lim_{n\rightarrow\infty} \frac{\#_w^j(n)}{n}
$$
\noindent as desired.
\end{proof}

\subsection{Symbolic dynamical systems}

\begin{proof}[Proof of Proposition \ref{prop:back_and_forth_induce}]
The two identities $\mu = \underline{\bar{\mu}}$,
and $\nu = \bar{\underline{\nu}}$ follow directly from the definitions. If $\mu$ is invariant, then for every cylinder $\cylinder{\word{w}}$, we have
$\bar{\mu}(\cylinder{\word{w}}) = \mu(\word{w}) = \sum_{a \in \Sigma}\mu(\word{w} \cdot a) = \sum_{a \in \Sigma}\bar{\mu}(\cylinder{\word{w} \cdot a})$; from this,
and the observation that $1 = \mu(\epsilon) = \bar{\mu}(\cylinder{\epsilon}) = \bar{\mu}(\Sigma^\omega)$,
 that $\bar{\mu}$ is a probability measure on $\Sigma$ with the sigma algebra generated by the cylinder sets. In addition,
as $\mu$ is invariant, we have for any cylinder $\cylinder{\word{w}}'$ that: 
$$
\bar{\mu}(S^{-1}(\cylinder{\word{w}})) =  \bar{\mu}\left( \bigcup_{a \in \Sigma} \cylinder{a \cdot \word{w}} \right) =  \sum_{a \in \Sigma}\bar{\mu}(\cylinder{a \cdot \word{w}}) = \sum_{a \in \Sigma} \mu(a \cdot \word{w}) = \mu(\word{w}) = \bar{\mu}(\cylinder{\word{w}})
$$
\noindent whence $\bar{\mu}$ is shift-invariant.

Conversely, if $\nu$ is a shift-invariant probability measure on $\Sigma^\omega$, we have for any $\word{w}$ that:
$$
\underline{\nu}(\word{w}) = \nu(\cylinder{\word{w}}) = \nu\left( \bigcup_{a \in \Sigma} \cylinder{\word{w} \cdot a} \right) =
\sum_{a \in \Sigma} \nu(\cylinder{\word{w} \cdot a}) = \sum_{a \in \Sigma} \underline{\nu}(\word{w} \cdot a)
$$
and
$$
\underline{\nu}(\word{w}) =  \nu(\cylinder{\word{w}}) = \nu(S^{-1}(\cylinder{\word{w}})) = \nu\left( \bigcup_{a \in \Sigma}\cylinder{a \cdot \word{w}} \right)
= \sum_{a \in \Sigma} \nu(\cylinder{a \cdot \word{w}}) =  \sum_{a \in \Sigma} \underline{\nu}(a \cdot \word{w})
$$
\noindent showing that $\underline{\nu}$ is invariant.
\end{proof}

\begin{proof}[Proof of Proposition \ref{prop:back_and_forth_final}]
For the first part, we prove \ref{it:1_1} $\Rightarrow$ \ref{it:1_2} $\Rightarrow$ \ref{it:1_3} $\Rightarrow$ \ref{it:1_1}. 

If there is a $\mu$-distributed  $\alpha \in \Sigma^\omega$, then for any $\word{w} \in \Sigma^*$ and any $\epsilon > 0$, for all sufficiently large $n$ we have $\sup_{b \in \Sigma \cup \{\lambda\}} \vert \countones{b \cdot \word{w}}{\alpha\vert_{\leq_n}}/n - \mu(\word{w}) \vert < \epsilon$.
Observe that every occurrence of a word on the form $a \cdot \word{w}$ in $\alpha$
contains an occurrence of $\word{w}$, and hence
$\countones{\word{w}}{\alpha\vert_{\leq_n}} \geq \sum_{a \in \Sigma} \countones{a \cdot \word{w}}{\alpha\vert_{\leq_n}}$. Conversely, for
every occurrence of $\word{w}$ starting at some position $i \geq 2$ in $\alpha$, there is exactly one $a \in \Sigma$ such that the word
$a \cdot \word{w}$ occurs at position $i-1$, whence $\countones{\word{w}}{\alpha\vert_{\leq_n}} \leq 1 + \sum_{a \in \Sigma} \countones{a \cdot \word{w}}{\alpha\vert_{\leq_n}}$, and hence:
\begin{align*}
\left\vert \mu(\word{w}) - \sum_{a \in \Sigma} \mu(a \cdot \word{w})\right\vert &=
\left\vert \mu(\word{w}) - \frac{\countones{\word{w}}{\alpha\vert_{\leq_n}}}{n}  + \frac{\countones{\word{w}}{\alpha\vert_{\leq_n}}}{n} - \sum_{a \in \Sigma} \mu(a \cdot \word{w})\right\vert \\
&\leq \left\vert \mu(\word{w}) - \frac{\countones{\word{w}}{\alpha\vert_{\leq_n}}}{n} \right\vert + \left\vert  \frac{\countones{\word{w}}{\alpha\vert_{\leq_n}}}{n} - \sum_{a \in \Sigma} \mu(a \cdot \word{w}) \right\vert \\
&< \epsilon + \frac{1}{n} + \left\vert \frac{\sum_{a \in \Sigma} \countones{a \cdot \word{w}}{\alpha\vert_{\leq_n}}}{n} - \sum_{a \in \Sigma} \mu(a \cdot \word{w}) \right\vert \\
&\leq \epsilon + \frac{1}{n} + \sum_{a \in \Sigma} \left\vert \frac{\countones{a \cdot \word{w}}{\alpha\vert_{\leq n}}}{n} - \mu(a \cdot \word{w})\right\vert\\
&< \epsilon + \frac{1}{n} + \vert \Sigma \vert \epsilon
\end{align*}
\noindent and as $\epsilon$ was arbitrary, we thus have $\mu(\word{w}) =  \sum_{a \in \Sigma} \mu(a \cdot \word{w})$. 
The case for
$\mu(\word{w}) = \sum_{a \in \Sigma} \mu(\word{w} \cdot a)$ is symmetric, \emph{mutatis mutandis}, and hence $\mu$ is invariant. If $\mu$ is invariant, then by Proposition \ref{prop:back_and_forth_induce},
$\bar{\mu}$ is a shift-invariant probability measure on $\Sigma^\omega$. If $\nu$ is a shift-invariant probability measure on $\Sigma^\omega$ such that   
$\bar{\mu} = \nu$, then by \cite[Main Thm.\ 2.1]{MadritchMance:constructing}, there exists $\alpha \in \Sigma^\omega$ generic for $\bar{\mu}$,
and thus for any admissible $\word{w} \in \Sigma^*$ 
$\lim_{n\rightarrow \infty} \frac{\countones{\word{w}}{\alpha\vert_{\leq n}}}{n} = \nu(\word{w}) = \bar{\mu}(\cylinder{\word{w}}) = \mu(\word{w})$. Observe that any inadmissible word
$w = a_1 \cdots a_n$ has $\prod_{i=1}^n \mu(a_i) = \mu(w) = 0$, whence $\mu(a_i) = 0$ for some $i$, and hence $\lim_{n\rightarrow \infty} \frac{\countones{\word{w}}{\alpha\vert_{\leq n}}}{n} \leq
\lim_{n \rightarrow \infty} \frac{\countones{a_i}}{\alpha\vert_{\leq n}}{n} = 0$.
Hence, $\alpha$ is $\mu$-distributed.

For the second part, we prove \ref{it:2_1} $\Rightarrow$ \ref{it:2_3} $\Rightarrow$ \ref{it:2_2} $\Rightarrow$ \ref{it:2_1}. 
Assume that $\alpha$ is generic for $\nu$. By construction, $\underline{\nu}$ is a probability map such that
 $\alpha$ is $\underline{\nu}$-distributed, and by the first part of the proposition, $\underline{\nu}$ is invariant, as desired.
 If $\underline{\nu}$ is an invariant probability map, then as any measurable $A$ can be written as a disjount union of cylinder sets,
 and as we for any cylinder $\cylinder{\word{w}}$ have $S^{-1}(\cylinder{w}) =  \cup_{a \in \Sigma} \cylinder{a \cdot \word{w}}$,
 we obtain 
 $$
 \nu(S^{-1}(\cylinder{w})) = \underline{\nu}(\cup_{a \in \Sigma} \cylinder{a \cdot \word{w}}) = \sum_{a \in \Sigma} \underline{\nu}(a \cdot \word{w}) = \underline{\nu}(\word{w}) = \nu(\cylinder{w})
$$
\noindent showing that $\nu$ is shift-invariant. Finally, if $\nu$ is shift-invariant, it follows from \cite[Main Thm.\ 2.1]{MadritchMance:constructing}, there exists $\alpha \in \Sigma^\omega$ generic for $\nu$, as desired.
\end{proof}

\end{document}